\documentclass[11pt, reqno]{amsart}

\usepackage{microtype}[draft]
\usepackage[scr=euler]{mathalfa}
\usepackage{stmaryrd}
\usepackage{tikz-cd}
\usepackage{amsmath, amssymb}
\usepackage{mathtools}

\usepackage[hyphens]{url}
\usepackage[breaklinks=true]{hyperref}

\theoremstyle{plain}
\newtheorem{theorem}{Theorem}[section]
\newtheorem{lemma}[theorem]{Lemma}
\newtheorem{proposition}[theorem]{Proposition}

\newtheorem*{claim*}{Claim}

\newtheorem*{fact*}{Fact} 

\theoremstyle{definition}

\newtheorem{remark}[theorem]{Remark}

\newcommand{\Set}{\mathbf{Set}}

\newcommand{\As}{\mathscr{A}}
\newcommand{\Bs}{\mathscr{B}}
\newcommand{\Cs}{\mathscr{C}}
\newcommand{\Rs}{\mathscr{R}}
\newcommand{\Sa}{\mathscr{S}_{a}}
\newcommand{\Sb}{\mathscr{S}_{b}}
\newcommand{\sg}{\sigma}
\newcommand{\RA}{R^{\As}}
\newcommand{\RB}{R^{\Bs}}
\newcommand{\RaA}{R_{\alpha}^{\As}}
\newcommand{\RbA}{R_{\beta}^{\As}}
\newcommand{\RaB}{R_{\alpha}^{\Bs}}
\newcommand{\IMP}{\; \Rightarrow \;}
\newcommand{\AND}{\, \wedge \,}
\newcommand{\CS}{\mathcal{R}(\sg)}
\newcommand{\CSI}{\mathcal{R}(\sgp)}
\newcommand{\CSp}{\mathcal{R}_{\star}(\sg)}

\newcommand{\preford}{\sqsubseteq}

\newcommand{\IFF}{\Longleftrightarrow}
\newcommand{\rarr}{\rightarrow}

\newcommand{\id}{\mathsf{id}}
\newcommand{\ie}{\textit{i.e.}~}
\newcommand{\Count}{\#}
\newcommand{\LL}{\mathcal{L}}

\newcommand{\Lk}{\mathcal{L}_k}
\newcommand{\Lc}{\mathcal{L}(\Count)}
\newcommand{\Lck}{\mathcal{L}_{k}(\Count)}
\newcommand{\ELk}{\exists\mathcal{L}_{k}}
\newcommand{\Lvk}{\mathcal{L}^{k}}
\newcommand{\Lvck}{\mathcal{L}^{k}(\Count)}
\newcommand{\ELvk}{\exists\mathcal{L}^{k}}
\newcommand{\eqLk}{\equiv^{\Lk}}
\newcommand{\eqELk}{\equiv^{\ELk}}
\newcommand{\eqLck}{\equiv^{\Lck}}
\newcommand{\eqLvk}{\equiv^{\Lvk}}
\newcommand{\eqELvk}{\equiv^{\ELvk}}
\newcommand{\eqLvck}{\equiv^{\Lvck}}
\newcommand{\KK}{\mathcal{K}}
\newcommand{\BB}{\mathcal{B}}
\newcommand{\DD}{\mathcal{D}}
\newcommand{\eqL}{\equiv^{\LL}}
\newcommand{\vphi}{\varphi}
\newcommand{\iffdef}{\;\; \stackrel{\Delta}{\IFF} \;\;}
\newcommand{\GG}{\mathsf{G}}
\newcommand{\Gk}{\mathcal{G}_k}
\newcommand{\Ck}{\mathbb{C}_{k}}
\newcommand{\Cl}{\mathbb{C}_{l}}
\newcommand{\Kl}{\mathsf{Kl}}
\newcommand{\epsA}{\varepsilon_{\As}}
\newcommand{\epsB}{\varepsilon_{\Bs}}
\newcommand{\eqak}{\rightleftarrows_{k}}
\newcommand{\eqbk}{\leftrightarrow_{k}}
\newcommand{\eqck}{\cong_{\Kl(\Ck)}}
\newcommand{\eqaCk}{\rightleftarrows_{k}^{\mathbb{C}}}
\newcommand{\eqbCk}{\leftrightarrow_{k}^{\mathbb{C}}}
\newcommand{\eqcCk}{\cong_{k}^{\mathbb{C}}}
\newcommand{\Alk}{A^{\leq k}}
\newcommand{\Blk}{B^{\leq k}}

\newcommand{\CC}{\mathcal{C}}
\newcommand{\Gd}{G}
\newcommand{\Ek}{\mathbb{E}_{k}}
\newcommand{\Ekp}{\mathbb{E}_{k}^{+}}
\newcommand{\Eomega}{\mathbb{E}_{\omega}}
\newcommand{\El}{\mathbb{E}_{l}}
\newcommand{\REk}{R^{\Ek \As}}
\newcommand{\Pk}{\mathbb{P}_{k}}
\newcommand{\Pkp}{\mathbb{P}_{k}^{+}}
\newcommand{\Mk}{\mathbb{M}_{k}}
\newcommand{\Mz}{\mathbb{M}_{0}}
\newcommand{\Mkp}{\mathbb{M}_{k'}}

\newcommand{\MLck}{\MLk(\Count)}
\newcommand{\eqEMk}{\equiv^{\exists\MLk}}
\newcommand{\eqMk}{\equiv^{\MLk}}
\newcommand{\eqMck}{\equiv^{\MLck}}

\newcommand{\pref}{\sqsubseteq}
\newcommand{\prefgt}{\sqsupseteq}
\newcommand{\Ralph}{R_{\alpha}}
\newcommand{\PMA}{P^{\Mk (\As, a)}}

\newcommand{\PMAp}{P^{\Mk (\As, a')}}
\newcommand{\PMBp}{P^{\Mk (\Bs, b')}}
\newcommand{\PMpA}{P^{\Mkp (\As, a)}}
\newcommand{\PMpB}{P^{\Mkp (\Bs, b)}}
\newcommand{\PA}{P^{\As}}
\newcommand{\PB}{P^{\Bs}}
\newcommand{\RMA}{R^{\Mk (\As, a)}}

\newcommand{\RMAp}{R^{\Mk (\As, a')}}
\newcommand{\RMBp}{R^{\Mk (\Bs, b')}}
\newcommand{\RMpA}{R^{\Mkp (\As, a)}}
\newcommand{\RMpB}{R^{\Mkp (\Bs, b)}}
\newcommand{\simord}{\preceq}
\newcommand{\Tk}{\mathbb{P}_k}
\newcommand{\Kcomp}{\bullet}
\newcommand{\nset}{\mathbf{n}}
\newcommand{\kset}{\mathbf{k}}

\newcommand{\Linf}{\mathcal{L}_{\infty,\omega}}
\newcommand{\MLk}{\mathcal{M}_{k}}
\newcommand{\lsem}{\llbracket}
\newcommand{\rsem}{\rrbracket}
\newcommand{\boxa}{\Box_{\alpha}}
\newcommand{\dia}{\Diamond_{\alpha}}

\newcommand{\elen}{\exists_{\leq n}}
\newcommand{\egen}{\exists_{\geq n}}
\newcommand{\eqaEk}{\rightleftarrows_{k}^{\mathbb{E}}}
\newcommand{\eqbEk}{\leftrightarrow_{k}^{\mathbb{E}}}
\newcommand{\eqcEk}{\cong_{k}^{\mathbb{E}}}
\newcommand{\eqaPk}{\rightleftarrows_{k}^{\mathbb{P}}}
\newcommand{\eqbPk}{\leftrightarrow_{k}^{\mathbb{P}}}
\newcommand{\eqcPk}{\cong_{k}^{\mathbb{P}}}
\newcommand{\eqaMk}{\rightleftarrows_{k}^{\mathbb{M}}}
\newcommand{\eqbMk}{\leftrightarrow_{k}^{\mathbb{M}}}
\newcommand{\eqcMk}{\cong_{k}^{\mathbb{M}}}
\newcommand{\eqcMkp}{\cong_{k'}^{\mathbb{M}}}

\newcommand{\dom}{\mathrm{dom}}

\newcommand{\gb}{\sim^{\mathsf{g}}}

\newcommand{\Fraisse}{Fra\"{i}ss\'{e}~}
\newcommand{\pow}{\mathscr{P}}
\newcommand{\comp}{{\uparrow}}
\newcommand{\adj}{\frown}
\newcommand{\da}{{\downarrow}}
\newcommand{\hgt}{\mathsf{ht}}
\newcommand{\td}{\mathsf{td}}
\newcommand{\tw}{\mathsf{tw}}
\newcommand{\md}{\mathsf{md}}
\newcommand{\Gf}{\mathcal{G}}
\newcommand{\Gsg}{\mathcal{G}_{\sg}}
\newcommand{\cnE}{\kappa^{\mathbb{E}}}
\newcommand{\cnP}{\kappa^{\mathbb{P}}}
\newcommand{\cnM}{\kappa^{\mathbb{M}}}
\newcommand{\lbfn}{\lambda}
\newcommand{\pth}{\mathsf{path}}
\newcommand{\labar}[1]{\overset{#1}{\to}}
\newcommand{\ve}{\varepsilon}
\newcommand{\Zp}{\mathbb{Z}^{+}}
\newcommand{\Comon}{\mathsf{Comon}}
\newcommand{\MM}{\mathbb{M}}
\newcommand{\PM}{\mathbb{P}}
\newcommand{\Momega}{\MM_{\omega}}
\newcommand{\es}{\varnothing}

\newcommand{\WABC}{\mathsf{W}(\As,\Bs)}

\newcommand{\cvr}{\prec}
\newcommand{\rcvr}{\succ}
\newcommand{\RT}{\mathcal{R}^{E}}
\newcommand{\RTk}{\mathcal{R}^{E}_{k}}
\newcommand{\RPk}{\mathcal{R}^{P}_{k}}
\newcommand{\RCk}{\mathcal{R}^{\mathbb{C}}_{k}}

\newcommand{\arEk}{\to_{k}^{\mathbb{E}}}
\newcommand{\arPk}{\to_{k}^{\mathbb{P}}}
\newcommand{\arMk}{\to_{k}^{\mathbb{M}}}
\newcommand{\arCk}{\to_{k}^{\mathbb{C}}}
\newcommand{\eqk}{\rightleftarrows_{k}}

\newcommand{\vv}{\vec{v}}
\newcommand{\va}{\vec{a}}
\newcommand{\QA}{Q^{\As}}
\newcommand{\QAp}{Q^{a}}
\newcommand{\QBp}{Q^{b}}
\newcommand{\QS}{Q^{\Sa}}
\newcommand{\MQAa}{MQ^{(\As,a)}}
\newcommand{\MQAb}{MQ^{(\As,b)}}
\newcommand{\rew}{\, \rightsquigarrow \,}
\newcommand{\FV}{\mathsf{FV}}
\newcommand{\redE}{\twoheadrightarrow^{\mathbb{E}}}
\newcommand{\redP}{\twoheadrightarrow^{\mathbb{P}}}
\newcommand{\redM}{\twoheadrightarrow^{\mathbb{M}}}

\newcommand{\Paths}{\mathsf{Paths}}
\newcommand{\Trees}{\mathsf{Trees}}
\newcommand{\Forests}{\mathsf{Forests}}
\newcommand{\embed}{\rightarrowtail}
\newcommand{\lastp}{\mathsf{last}_{p}}
\newcommand{\lastpp}{\mathsf{last}_{p'}}
\newcommand{\ww}{\vec{w}}
\newcommand{\lt}[1]{\stackrel{#1}{\longrightarrow}}
\newcommand{\sgp}{\sigma^{+}}
\newcommand{\EM}{\mathsf{EM}}

\title[Relating Structure and Power]{Relating Structure and Power: \\ Comonadic Semantics for Computational Resources}

\author{Samson Abramsky}
\address{Department of Computer Science, University of Oxford, Wolfson Building, Parks Road, Oxford OX1 3QD, U.K.}\email{samson.abramsky@cs.ox.ac.uk}

\author{Nihil Shah}
\address{Department of Computer Science, University of Oxford, Wolfson Building, Parks Road, Oxford OX1 3QD, U.K.}\email{nihil.shah@cs.ox.ac.uk}


\begin{document}

\maketitle

\begin{abstract}
Combinatorial games are widely used in finite model theory, constraint satisfaction, modal logic and concurrency theory to characterize logical equivalences between structures. 
In particular, Ehrenfeucht-\Fraisse games, pebble games, and bisimulation games play a central role. 
We show how each of these types of games can be described in terms of an indexed family of comonads on the category of relational structures and homomorphisms. The index $k$ is a resource parameter which bounds the degree of access to the underlying structure. 
The coKleisli categories for these comonads can be used to give syntax-free characterizations of a wide range of important logical equivalences. Moreover, the coalgebras for these indexed comonads can be used to characterize key combinatorial parameters: tree-depth for the Ehrenfeucht-\Fraisse comonad, tree-width for the pebbling comonad, and synchronization-tree depth for the modal unfolding comonad. These results pave the way for systematic connections between two major branches of the field of logic in computer science which hitherto have been almost disjoint: categorical semantics, and finite and algorithmic model theory.
\end{abstract}

\section{Introduction}

There is a remarkable divide in the field of logic in Computer Science, between two distinct strands: one focusing on semantics and compositionality (``Structure''), the other on expressiveness and complexity (``Power''). 
It is remarkable because these two fundamental aspects of our field are studied using almost disjoint technical languages and methods, by almost disjoint research communities.
We believe that bridging this divide is a major issue in Computer Science, and may hold the key to fundamental advances in the field.

In this paper, we develop a novel approach to relating categorical semantics, which exemplifies the first strand, to finite model theory, which exemplifies the second. It builds on the ideas introduced in \cite{abramsky2017pebbling}, but goes much further, showing clearly that there is a strong and robust connection, which can serve as a basis for many further developments. These developments were first introduced in the conference version of this paper \cite{DBLP:conf/csl/AbramskyS18}. In this journal version, we have added more detailed proofs, and several additional sections demonstrating how game comonads generalize constructions in finite model theory. We have also introduced a new and more satisfactory unified account of back-and-forth equivalences, in terms of open pathwise embeddings, which capture a general notion of \emph{property-preserving bisimulation}.

\subsection*{Note on the presentation}
Since this work aims at connecting what are at present largely disjoint research areas and communities, we have attempted to provide enough background exposition in the paper to make it reasonably self-contained.
In the following section, we provide some background on model comparison games, and their correspondence with logical equivalences, to assist readers mainly coming from the categorical semantics side. For those whose background is mainly in finite model theory and descriptive complexity, we assume only some prior exposure to the basic notions of category, functor, and natural transformation.
Notions relating to comonads are presented in a self-contained way.

We thank the journal referees for their encouragement of, and indeed insistence on, our efforts in this regard.

\subsection*{Acknowledgements} 
The authors are grateful for the support of EPSRC grant EP/T00696X/1, and the 
Department of Computer Science, University of Oxford. They benefited from
very helpful comments from Dan Marsden.
They also acknowledge the valuable comments of the journal referees. In particular, they appreciate the detailed comments made by the first reviewer through several rounds of revisions, which led to a number of improvements in the presentation. 

\subsection*{The setting}
Relational structures and the homomorphisms between them play a fundamental r\^{o}le in finite model theory, constraint satisfaction and database theory. The existence of a homomorphism $A \rarr B$ is an equivalent formulation of constraint satisfaction, and also equivalent to the preservation of existential positive sentences \cite{chandra1977optimal}. This setting also generalizes what has become a central perspective in graph theory \cite{hell2004graphs}.
\subsection*{Model theory and deception}
In a sense, the purpose of model theory is ``deception''.  It allows us to see structures not ``as they really are'', \ie up to isomorphism, but only up to \emph{definable properties}, where definability is relative to a logical language $\LL$. The key notion is \emph{logical equivalence} $\eqL$. Given structures $\As$, $\Bs$ over the same vocabulary:
\[ \As \eqL \Bs \iffdef \forall \vphi \in \LL. \; \As \models \vphi \; \IFF \; \Bs \models \vphi . \]
If a class of structures $\KK$ is definable in $\LL$, then it must be saturated under $\eqL$. Moreover, for a wide class of cases of interest in finite model theory, the converse holds \cite{kolaitis1992infinitary}.

The idea of syntax-independent characterizations of logical equivalence is quite a classical one in model theory, exemplified by the Keisler-Shelah theorem \cite{shelah1971every}.
It acquires additional significance in finite model theory, where model comparison games such as Ehrenfeucht-\Fraisse games, pebble games and bisimulation games play a central role \cite{Libkin2004}.

We offer a new perspective on these ideas. We shall study these games, not as external artefacts, but as semantic constructions in their own right.
Each model-theoretic comparison game encodes ``deception'' in terms of limited access to the structure. These limitations are indexed by a parameter which quantifies the resources which control this access. For Ehrenfeucht-\Fraisse games and bisimulation games, this is the number of rounds; for pebble games, the number of pebbles.

\subsection*{Main Results}
We now give a conceptual overview of our main results. Technical details will be provided in the following sections. 

We shall consider three forms of model comparison game: Ehrenfeucht-\Fraisse games, pebble games and bisimulation games \cite{Libkin2004}.
For each of these notions of game $\GG$, and value of the resource parameter $k$, we shall define a corresponding \emph{comonad} $\Ck$ on the category of relational structures and homomorphisms  over some relational vocabulary. For each structure $\As$, $\Ck \As$ is another structure over the same vocabulary, which encodes the limited access to $\As$ afforded by playing the game on $\As$ with $k$ resources. There is always an associated homomorphism $\epsA : \Ck \As \rarr \As$ (the \emph{counit} of the comonad), so that $\Ck \As$ ``covers'' $\As$. Moreover, given a homomorphism $h : \Ck \As \rarr \Bs$, there is a \emph{Kleisli coextension} homomorphism $h^* : \Ck \As \rarr \Ck \Bs$. This allows us to form the \emph{coKleisli category} $\Kl(\Ck)$ for the comonad. The objects are relational structures, while the morphisms from $\As$ to $\Bs$ in $\Kl(\Ck)$ are exactly the homomorphisms of the form $\Ck \As \rarr \Bs$. Composition of these morphisms uses the Kleisli coextension. The connection between this construction and the corresponding form of game $\GG$ is expressed  by the following result:
\begin{theorem}
The following are equivalent:
\begin{enumerate}
\item There is a coKleisli morphism $\Ck \As \rarr \Bs$
\item Duplicator has a winning strategy for the existential $\GG$-game with $k$ resources, played from $\As$ to $\Bs$.
\end{enumerate}
\end{theorem}
The existential form of the game has only a ``forth'' aspect, without the ``back''. This means that Spoiler can only play in $\As$, while Duplicator only plays in $\Bs$. This corresponds to the asymmetric form of the coKleisli morphisms $\Ck \As \rarr \Bs$. Intuitively, Spoiler plays in $\Ck \As$, which gives them limited access to $\As$, while Duplicator plays in $\Bs$. The Kleisli coextension guarantees that Duplicator's strategies can always be lifted to $\Ck \Bs$; while we can always compose a strategy $\Ck \As \rarr \Ck \Bs$ with the counit on $\Bs$ to obtain a coKleisli morphism.

This asymmetric form may seem to limit the scope of this approach, but in fact this is not the case. For each of these comonads $\Ck$, we have the following equivalences:
\begin{itemize}
\item $\As \eqak \Bs$ iff there are coKleisli morphisms $\Ck \As \rarr \Bs$ and $\Ck \Bs \rarr \As$. Note that there need be no relationship between these morphisms.
\item $\As \eqck \Bs$ iff $\As$ and $\Bs$ are isomorphic in the coKleisli category $\Kl(\Ck)$. This means that there are morphisms $\Ck \As \rarr \Bs$ and $\Ck \Bs \rarr \As$ which are inverses of each other in $\Kl(\Ck)$.
\end{itemize}
Clearly, $\eqck$ strictly implies $\eqak$.
We can also define an intermediate ``back-and-forth'' equivalence $\eqbk$, in terms of the existence of a suitable kind of span of coKleisli morphisms.

For each of our three types of game, there are corresponding  fragments $\Lk$ of first-order logic:
\begin{itemize}
\item For Ehrenfeucht-\Fraisse games, $\Lk$ is the fragment of quantifier-rank $\leq k$.
\item For pebble games, $\Lk$ is the $k$-variable fragment.
\item For bisimulation games over relational vocabularies with symbols of arity at most 2, $\Lk$ is the modal fragment \cite{andreka1998modal} with modal depth $\leq k$.
\end{itemize}
In each case, we write $\ELk$ for the existential positive fragment of $\Lk$, and $\Lck$ for the extension of $\Lk$ with counting quantifiers \cite{Libkin2004}.

We can now state our first main result, in a suitably generic form.
\begin{theorem}
For finite structures $\As$ and $\Bs$:

\begin{tabular}{llcl}
(1) & $\As \eqELk \Bs$ & $\; \IFF \;$ & $\As \eqak \Bs$. \\
(2) & $\As \eqLk \Bs$ & $\; \IFF \;$ & $\As \eqbk \Bs$. \\
(3) & $\As \eqLck \Bs$ & $\; \IFF \;$ & $\As \eqck \Bs$.
\end{tabular}
\end{theorem}
Note that this is really a family of three theorems, one for each type of game $\GG$. Thus in each case, we capture the salient logical equivalences in syntax-free, categorical form.

We now turn to the significance of indexing by the resource parameter $k$. When $k \leq l$, we have a natural inclusion morphism $\Ck \As \rarr \Cl \As$, since playing with $k$ resources is a special case of playing with $l \geq k$ resources. This tells us that the smaller $k$ is, the easier it is to find a morphism $\Ck \As \rarr \Bs$. Intuitively, the more we restrict Spoiler's abilities to access the structure of $\As$, the easier it is for Duplicator to win the game.

The contrary analysis applies to morphisms $\As \rarr \Ck \Bs$. The smaller $k$ is, the \emph{harder} it is find such a morphism. Note, however, that if $\As$ is a finite structure of cardinality $k$, then $\As \eqak \Ck \As$. In this case, with $k$ resources we can access the whole of $\As$. What can we say when $k$ is strictly smaller than the cardinality of $A$?

It turns out that there is a beautiful connection between these indexed comonads and combinatorial invariants of structures. This is mediated by the notion of \emph{coalgebra}, another fundamental (and completely general) aspect of comonads. A coalgebra for a comonad $\Ck$ on a structure $\As$ is a morphism $\As \rarr \Ck \As$ satisfying certain properties. We define the \emph{coalgebra number} of a structure $\As$, with respect to the indexed family of comonads $\Ck$, to be the least $k$ such that there is a $\Ck$-coalgebra  on $\As$.

We now come to our second main result.
\begin{theorem}
\begin{enumerate}
\item For the Ehrenfeucht-\Fraisse comonad, the coalgebra number of $\As$ corresponds precisely to the \emph{tree-depth} of $\As$ \cite{nevsetvril2006tree}.
\item For the pebbling comonad, the coalgebra number of $\As$ corresponds precisely to the \emph{tree-width} of $\As$.
\item For the modal comonad, the coalgebra number of $\As$ corresponds precisely to the \emph{modal unfolding depth} of $\As$. 
\end{enumerate}
\label{thm:introCoalgebra}
\end{theorem}
The main idea behind these results is that coalgebras on $\As$ are in bijective correspondence with decompositions of $\As$ of the appropriate form. We thus obtain categorical characterizations of these key combinatorial parameters.

These two results appeared in the first version of this paper \cite{DBLP:conf/csl/AbramskyS18}, and  unify logical resources with combinatorial resources. In this expanded version of the paper, we provide two new families of results that can be seen as further aspects of this unification. The first result utilizes a general fact about comonads to link the logical resource characterized by coKleisli morphisms $\Ck \As \rightarrow \Bs$ to the combinatorial resource characterized by the coalgebas $\As \rightarrow \Ck\As$. 
For each game comonad $\Ck$, we can define a preorder $\arCk$ over the category of structures, where
$\As \arCk \Bs$ iff for all structures $\Cs$ with combinatorial resource parameter $\leq k$,  
$\Cs \rarr \As \IMP \Cs \rarr \Bs$.\footnote{We use the notation $\Cs \rarr \As$ to mean that there exists a morphism with domain $\Cs$ and codomain $\As$.}
This preorder, when instantiated to the case of Ehrenfeucht-{\Fraisse}games and tree-depth, is used extensively in \cite{Rossman2008}.
We demonstrate that this preorder is characterized by the coKleisli category of $\Ck$.
\begin{proposition}
$\As \arCk \Bs \; \IFF \; \Ck \As \rightarrow \Bs$.
\end{proposition}

For the second result, we show how coalgebras can be interpreted syntactically. We show that a coalgebra $\As \rarr \Ck\As$ encodes a witness  to the fact that the canonical conjunctive query $\QA$ of $\As$ is equivalent to a sentence in the logical fragment corresponding to the comonad.
\begin{theorem}
  $\As$ has a $\Ck$-coalgebra iff $\QA$ can be rewritten into a formula $\varphi$, where $\varphi$ has logical resource at most $k$.
\end{theorem}
The formula $\varphi$ can be read out from the coalgebra map.

We also show how to rephrase the coalgebra results of Theorem~\ref{thm:introCoalgebra} in the language of adjunctions. 
\begin{theorem}
Each game comonad arises from an adjunction between the category of structures and a category of tree-ordered structures:
\begin{enumerate}
  \item For the Ehrenfeucht-{\Fraisse} comonad, the associated category is $k$-height tree-ordered structures.
  \item For the pebbling comonad, the associated category is $k$-pebble tree-ordered structures. 
  \item For the modal comonad, the associated category is $k$-height tree-ordered Kripke structures. 
\end{enumerate}
In each case, the left adjoint is the evident forgetful functor which forgets the order. The right adjoint is given by the functor part of the comonad.
\end{theorem}

We have also put the notion of $I$-morphism introduced in \cite{abramsky2017pebbling}, and used extensively in \cite{DBLP:conf/csl/AbramskyS18}, on a more systematic basis, in terms of \emph{relative comonads} \cite{altenkirch2010monads}.

Finally, an important contribution of this expanded version of the paper is to give a unified, general treatment of back-and-forth equivalences, which play a central role in finite model theory. To do this, we use a refined version of the well-known notion of open map bisimulation \cite{joyal1993bisimulation}, in which the maps in the span witnessing a bisimulation are required not only to be open, but to be \emph{pathwise embeddings}.
This supports a general notion of \emph{property-preserving bisimulation}, which specialises to yield all the back-and-forth equivalences we study here, as well as many other examples, e.g.~\emph{guarded bisimulation} \cite{abramsky2017pebbling}.

\section{Background}

We shall now review some background on logics, model comparison games, and comonads.

\subsection{Notational preliminaries}
A relational vocabulary $\sg$ is a set of relation symbols $R$, each with a specified positive integer arity.
A $\sg$-structure $\As$ is given by a set $A$, the universe of the structure, and for each $R$ in $\sg$ with arity $n$, a relation $\RA \subseteq A^n$. A homomorphism $h : \As \rarr \Bs$ is a function $h : A \rarr B$ such that, for each relation symbol $R$ of arity $n$ in $\sg$, for all $a_1, \ldots , a_n$ in $A$:
$\RA(a_1,\ldots , a_n) \IMP \RB(h(a_1), \ldots , h(a_n))$. We write $\CS$ for the category of $\sg$-structures and homomorphisms.
 
A homomorphism $h : \As \to \Bs$ is \emph{strong} if for $n$-ary $R$ in $\sg$ and $a_1, \ldots , a_n \in A$, $R^{\Bs}(h(a_1), \ldots , h(a_n)) \Rightarrow R^{\As}(a_1, \ldots , a_n)$.
An \emph{embedding}  of relational structures is an injective strong homomorphism.  It is easily verified that the embeddings in $\CS$ are exactly the extremal monomorphisms; those monomorphisms $m$ such that, whenever $m = f \circ e$ for an epimorphism $e$, $e$ is an isomorphism.

A \emph{partial homomorphism} from $\As$ to $\Bs$ is a finite relation $r \subseteq A \times B$ which is single-valued on its domain $A_0 := \pi_1(r)$, \ie~$(a, b), (a, b') \in r \IMP b = b'$, and which is a homomorphism on the substructure of $\As$ determined by $A_0$.
A \emph{partial isomorphism} from $\As$ to $\Bs$ is a partial homomorphism $r$ such that the converse relation $r^{c} \subseteq B \times A$ is a partial homomorphism from $\Bs$ to $\As$.

We shall write $\Alk$ for the set of non-empty sequences of length $\leq k$ on a set $A$. We use list notation $[a_1,\ldots , a_j]$ for such sequences, and indicate concatenation by juxtaposition, i.e. $ss'$ is the concatenation of the sequences $s$ and $s'$. We write $s \pref t$ for the prefix ordering on sequences. If $s \preford t$, there is a unique $s'$ such that $ss' = t$, which we refer to as the suffix of $s$ in $t$. For each positive integer $n$, we define $\nset := \{ 1, \ldots , n\}$. 

We shall need a few notions on posets. The comparability relation on a poset  $(P, {\leq})$ is $x \comp y$ iff  $x \leq y$ or $y \leq x$.  A chain in a poset $(P, {\leq})$ is  a subset $C \subseteq P$ such that, for all $x, y \in C$, $x \comp y$. A \emph{forest} is a poset $(F, {\leq})$ such that, for all $x \in F$, the  set of predecessors $\da(x) \, := \, \{ y \in F \mid y \leq x\}$ is a finite chain. The height $\hgt(F)$ of a forest $F$ is $\sup_{C} | C |$, where $C$ ranges over chains in $F$.
A \emph{tree} is a forest with a least element $\bot$ (the root).
We write the covering relation for a poset as $\cvr$; thus $x \cvr y$ iff $x \leq y$, $x \neq y$, and for all $z$, $x \leq z \leq y$ implies $z = x$ or $z = y$.

\subsection{Logic fragments}
\label{fragssec}

We shall be considering logics $\LL$ which arise as fragments of $\Linf$, the extension of first-order logic with infinitary conjunctions and disjunctions, but where formulas contain only finitely many variables. In particular, we will consider the fragments $\Lk$, of formulas with quantifier rank $\leq k$, and $\Lvk$, the $k$-variable fragment. These play a fundamental r\^ole in finite model theory.

We shall also consider two variants for each of these fragments $\LL$. One is the existential positive fragment $\exists\LL$, which contains only those formulas of $\LL$ built using existential quantifiers, conjunction and disjunction. The other is $\Lc$, the extension of $\LL$ with counting quantifiers. These have the form $\elen$, $\egen$, where the semantics of $\As \models \egen x. \, \psi$ is that there exist at least $n$ distinct elements of $A$ satisfying $\psi$.

For modal logic, we will  consider $\MLk$, the \emph{modal fragment} of modal depth $\leq k$. This arises from the standard translation of (multi)modal logic into $\Linf$ \cite{blackburn2002modal}. Let us fix a relational vocabulary $\sg$ with symbols of arity $\leq 2$. For each unary symbol $P$, there will be a corresponding propositional variable $p$. Formulas are built from these propositional variables by propositional connectives, and modalities $\boxa$, $\dia$ corresponding to 
relations $\Ralph$  indexed by the binary relation symbols from the given alphabet.
Modal formulas $\vphi$ then admit a translation into formulas $\lsem \vphi \rsem = \psi(x)$ in one free variable. The translation sends propositional variables $p$ to $P(x)$, commutes with the propositional connectives, and  sends $\dia \vphi$ to $\exists y. \, \Ralph(x,y) \AND \psi(y)$, and $\boxa \vphi$ to $\forall y. \, \Ralph(x, y) \rightarrow \psi(y)$, where $\psi(x) = \lsem \vphi \rsem$. This translation is semantics-preserving: given  a $\sg$-structure $\As$ and $a \in A$, then $\As, a \models \vphi$ in the sense of Kripke semantics iff $\As \models \psi(a)$ in the standard model-theoretic sense, where $\psi(x) = \lsem \vphi \rsem$.

We define the modal depth of a modal formula $\vphi$ as the maximum nesting depth of modalities occurring in $\vphi$. $\MLk$ is then the image of the translation of modal formulas of modal depth $\leq k$. The existential positive fragment $\exists\MLk$ arises from the modal sublanguage in which formulas are built from propositional variables using only conjunction, disjunction and the diamond modalities $\dia$.

Extensions of the modal language with counting capabilities have been studied in the form of \emph{graded modalities} \cite{Rijke2000}. These have the form $\dia^n$, $\boxa^n$, where $\As, a \models \dia^n \vphi$ if there are at least $n$ $\Ralph$-successors of $a$ which satisfy $\vphi$. 
We define $\MLk(\Count)$ to be the extension of the modal fragment with graded modalities.

\subsection{Model comparison games}

We review the model comparison games we will be concerned with in this paper. We are given structures $\As$ and $\Bs$.
Each game is a 2-person game, played between Spoiler, who is trying show that the structures are different, and Duplicator, who is trying to show they are the same. Each game is played in a number of rounds:
\begin{itemize}
    \item In the Ehrenfeucht-\Fraisse game, in the $i$'th round Spoiler chooses an element in one of the structures, and Duplicator then chooses an element in the other structure. Thus after $k$ rounds have been played, we have sequences $[a_1, \ldots , a_k]$ and $[b_1, \ldots b_k]$ of elements from $\As$ and $\Bs$ respectively. Duplicator wins  this play if $r := \{ (a_i,b_i) \mid 1 \leq i \leq k \}$ is a partial isomorphism. Duplicator wins the $k$-round game if they have a winning strategy in the usual sense. The resource parameter here is $k$, the number of rounds.
    
    \item In the pebble game, each player has $k$ pebbles available, which can be ``placed'' on elements of either structure. We label the pebbles by integers in $\kset$. In the $i$'th round, Spoiler places some pebble $p_i$ on an element of one of the structures. Note that if $p_i$ was previously placed on some other element, the effect is to remove it from that element, and place it on the newly chosen element.
    Duplicator then places their corresponding pebble $p_i$ on an element of the other structure.
     Thus the current positions of the pebbles determine  ``windows'' of  size bounded by $k$ onto the structures. These windows can slide around over different parts of the structures as moves are played.
     
    After $n$ rounds of the game have been played, we have sequences $[(p_1,a_1), \ldots , (p_n, a_n)]$ and $[(p_1, b_1), \ldots (p_n, b_n)]$ which record the placings of pebbles on elements during the play. The \emph{current placing} of pebble $p$ is the last element in the sequence with first component $p$.  Duplicator wins this play if the relation $r$ determined by the current placings of the pebbles is a partial isomorphism. Duplicator wins the $k$-pebble game if they have a strategy which is winning after $n$ rounds, for all $n$. The resource parameter here is $k$, the number of pebbles.
    
    \item Finally, in the bisimulation game, we have structures over a relational signature with symbols of arity at most 2. The game is played between ``pointed structures'' $(\As, a)$, $(\Bs, b)$, with specified elements $a \in A$ and $b \in B$. We start with elements $(a_0, b_0)$, where $a_0 = a$ and $b_0 = b$. At each round $i+1$, where the current elements are $(a_i, b_i)$, Spoiler chooses one of the structures, e.g.~$\As$, one of the binary relations $R_i$, and an element $a_{i+1}$ such that $R_i^{\As}(a_i, a_{i+1})$. Duplicator must respond in the other structure, in our example in $\Bs$ with $b_{i+1}$, such that $R_i^{\Bs}(b_i, b_{i+1})$. If Duplicator has no such response available, they lose. Duplicator wins after $k$ rounds if, for all unary predicates $P$ in the signature, we have $P^{\As}(a_{i}) \IFF P^{\Bs}(b_i)$ for all $i$ with $0 \leq i \leq k$. The resource parameter here is the number of rounds $k$.
\end{itemize}
There are classical theorems (Ehrenfeucht-\Fraisse \cite{fraisse1955quelques,ehrenfeucht1961application}, Kolaitis-Vardi \cite{kolaitis1992infinitary}, Hennessy-Milner \cite{hennessy1980observing}) which characterize various logical equivalences in terms of these games. These can be summarized as follows:
\begin{theorem}
\label{classicgamesthm}
Let $\As$ and $\Bs$ be structures over a finite vocabulary.\footnote{In the modal case, this means that there  finite numbers both of modalities, and of propositional atoms.}
\begin{enumerate}
    \item Duplicator has a winning strategy in the $k$-round Ehrenfeucht-\Fraisse game from $\As$ to $\Bs$ iff $\As$ and $\Bs$ satisfy the same sentences of quantifier rank $\leq k$.
    
    \item Duplicator has a winning strategy in the $k$-pebble game from $\As$ to $\Bs$ iff $\As$ and $\Bs$ satisfy the same sentences of $k$-variable logic.
    
    \item Duplicator has a winning strategy in the $k$-round bisimulation game from $(\As, a)$ to $(\Bs, b)$ iff $(\As, a)$ and $(\Bs, b)$ satisfy the same modal formulas of depth $\leq k$.
\end{enumerate}
\end{theorem}
We refer to \cite{Libkin2004,ebbinghaus2005finite,blackburn2002modal} for textbook accounts and extensive bibliographies.

The above games are symmetric or ``back-and-forth'', since Spoiler can choose which structure to play in at each round.
There are also asymmetric (forth only) or \emph{existential} versions, in which Spoiler always plays in $\As$, and Duplicator always responds in $\Bs$. The winning conditions are modified for these existential versions:
\begin{itemize}
    \item For the existential Ehrenfeucht-\Fraisse game, the winning condition is that $r$ is a partial homomorphism.
    \item Similarly for the existential $k$-pebble game.
    \item For the $k$-round \emph{simulation game}, as the existential version of the bisimulation game is called, the winning condition is that, for all unary predicates $P$, $P(a_i) \IMP P(b_i)$ for all $i$ with $0 \leq i \leq k$.
\end{itemize}
For these existential games, we obtain a corresponding version of Theorem~\ref{classicgamesthm}.
\begin{theorem}
\label{classicposgamesthm}
Let $\As$ and $\Bs$ be structures over a finite vocabulary.
\begin{enumerate}
    \item Duplicator has a winning strategy in the existential $k$-round Ehrenfeucht-\Fraisse game from $\As$ to $\Bs$ iff for every existential-positive sentence $\vphi$ of quantifier rank $\leq k$, $\As \models \vphi$ implies $\Bs \models \vphi$. 
    
    \item Duplicator has a winning strategy in the existential $k$-pebble game from $\As$ to $\Bs$ iff for every existential-positive sentence $\vphi$ of $k$-variable logic, $\As \models \vphi$ implies $\Bs \models \vphi$.
    
    \item Duplicator has a winning strategy in the $k$-round simulation game from $(\As, a)$ to $(\Bs, b)$ iff for every existential-positive sentence $\vphi$ of the modal fragment with modal depth $\leq k$, $\As \models \vphi$ implies $\Bs \models \vphi$.
\end{enumerate}
\end{theorem}

\subsection{Comonads}
\label{comonadsec}

We recall that a comonad $(G, \varepsilon, \delta)$ on a category $\CC$ is given by a functor $G : \CC \rarr \CC$, and natural transformations $\varepsilon : G \Rightarrow I$ (the counit), and $\delta : G \Rightarrow G^2$ (the comultiplication), subject to the conditions that the following diagrams commute, for all objects $A$ of $\CC$:
\begin{center}
\begin{tikzcd}
G A \ar[r, "\delta_A"]  \ar[d, "\delta_A"']
& G G A \ar[d, "G \delta_A"] \\
G G A \ar[r, "\delta_{G A}"]  
& G G G A
\end{tikzcd}
$\qquad \qquad$
\begin{tikzcd}
G A \ar[r, "\delta_A"]  \ar[rd,equal] \ar[d, "\delta_A"'] 
& G G A \ar[d, "G \epsA"] \\
G G A \ar[r, "\varepsilon_{GA}"] 
& G A
\end{tikzcd}
\end{center}

An equivalent formulation is  \emph{comonad in Kleisli form} \cite{manes2012algebraic}. This is given by an object map $G$, arrows $\varepsilon_A : GA \rarr A$ for every object $A$ of $\CC$, and a Kleisli coextension operation which takes $f : GA \rarr B$ to $f^* : GA \rarr GB$. These must satisfy the following equations:
\[ \varepsilon_{A}^* = \id_{\Gd A}, \qquad \varepsilon_{B} \circ f^* = f, \qquad (g \circ f^*)^* = g^* \circ f^* , \]
where $g : GB \to C$.
We can then extend $G$ to a functor by $\Gd f = (f \circ \varepsilon_{A})^*$; and if we define the comultiplication $\delta : \Gd \Rightarrow \Gd^2$ by $\delta_{A} = \id_{GA}^*$, then $(\Gd, \varepsilon, \delta)$ is a comonad in the standard sense.
Conversely, given a comonad $(\Gd, \varepsilon, \delta)$, we can define the coextension by $f^* = Gf \circ \delta_A$.
This allows us to define the coKleisli category $\Kl(G)$, with objects the same as those of $\CC$, and morphisms from $A$ to $B$ given by the morphisms in $\CC$ of the form $GA \rarr B$. Kleisli composition of $f : GA \rarr B$ with $g : GB \rarr C$ is given by $g \Kcomp f \, := \, g \circ f^*$.

\section{Game Comonads}

We shall now define the three comonads we shall study, corresponding to the three forms of model comparison game described in the previous section.

\subsection{The Ehrenfeucht-\Fraisse Comonad}

We shall define a comonad $\Ek$ on $\CS$ for each positive integer $k$. It will be convenient to define $\Ek$ in Kleisli form.
For each structure $\As$, we define a new structure $\Ek \As$, with universe $\Ek A \, := \, \Alk$. We define the map $\epsA : \Ek A \rarr A$ by $\epsA [a_1, \ldots , a_j ] = a_j$. For each relation symbol $R$ of arity $n$, we define $\REk$ to be the set of $n$-tuples $(s_1, \ldots , s_n)$ of sequences which are pairwise comparable in the prefix ordering, and such that $\RA(\epsA s_1, \ldots , \epsA s_n)$. Finally, we define the coextension. Given a homomorphism $f : \Ek \As \rarr \Bs$, we define $f^* : \Alk \rarr \Blk$ by $f^* [a_1, \ldots , a_j ] = [b_1, \ldots , b_j]$, where $b_i = f [a_1, \ldots , a_i]$, $1 \leq i \leq j$.

\begin{proposition}\label{efcomonprop}
The triple $(\Ek, \varepsilon, (\cdot)^*)$ is a comonad in Kleisli form.
\end{proposition}
\begin{proof}
It is well known that non-empty lists form a comonad on $\Set$ \cite{ahman2012container}. This is easily adapted to showing that non-empty lists of length $\leq k$ form a comonad. To lift this to $\CS$, we must show that each $\varepsilon_{\As}$ is a homomorphism, and that if $f : \Ek \As \to \Bs$ is a homomorphism, so is $f^*$. The first follows directly from the definition of $\REk$ for each $R$ in $\sg$.
For the second, if $\REk(s_1,\ldots,s_n)$, then $s_i \comp s_j$ implies $f^*(s_i) \comp f^*(s_j)$, and since $f$ is a homomorphism, $\REk( s_1, \ldots , s_n)$ implies $\RB(\varepsilon f^*(s_1),\ldots , \varepsilon f^*(s_n))$, so $R^{\Ek \Bs}(f^*(s_1),\ldots , f^*(s_n))$ as required.
\end{proof}
Intuitively, an element of $\Alk$ represents a play in $\As$ of length $\leq k$. A coKleisli morphism $\Ek \As \rarr \Bs$ represents a Duplicator strategy for the existential Ehrenfeucht-\Fraisse game with $k$ rounds, where Spoiler plays only in $\As$, and $b_i = f [a_1, \ldots , a_i]$ represents Duplicator's response in $\Bs$ to the $i$'th move by Spoiler. The winning condition for Duplicator in this game is that, after $k$ rounds have been played, 
the \textit{induced relation} $\{ (a_i, b_i) \mid 1 \leq i \leq k \}$ is a partial homomorphism from $\As$ to $\Bs$. 

These intuitions are confirmed by the following result.
\begin{theorem}
\label{EFgamethm}
The following are equivalent:
\begin{enumerate}
\item There is a homomorphism $f : \Ek \As \rarr \Bs$.
\item Duplicator has a winning strategy for the existential Ehrenfeucht-\Fraisse game with $k$ rounds, played from $\As$ to $\Bs$.
\end{enumerate}
\end{theorem}
\begin{proof}
 $(1) \Rightarrow (2)$. Suppose Spoiler plays $a_{1},\dots,a_{k}$, then for each $i \in \{1,\dots,k\}$, Duplicator responds with $b_i = b_j$ if $a_j = a_i$ for some $j < i$ \footnote{We could eliminate this case by imposing an additional \textit{I-morphism} condition on the homomorphisms of type $\Ek \As \rarr \Bs$. This condition is discussed in Section~\ref{Imorsec}.} or $b_i = f[a_1,\dots,a_i]$ otherwise. We must show that $\gamma = \{(a_{i},b_{i}) \mid 1 \leq i \leq k \}$ is a partial homomorphism. By construction, $b_i = b_j$ if $a_i = a_j$, so $\gamma$ is a well-defined partial function.  Let $R \in \sg$ and suppose $R^{\As}(a_{i_1},\dots,a_{i_n})$ with $i_{j}  \in \{1,\dots,k\}$ for all $j \leq n$. Let $s_{i_{j}}$ be the minimal subsequence of $s = [a_1,\dots, a_k]$ that ends in $a_{i_{j}}$. Since each $s_{i_{j}}$ is a subsequence of $s$, each pair of $s_{i_{j}}$ is prefix comparable and by definition, $\epsA(s_{i_{j}}) = a_{i_{j}}$, so $R^{\Ek\As}(s_{i_{1}},\dots,s_{i_{n}})$. By $f$ being a homomorphism, $R^{\Bs}(b_{i_{1}},\dots,b_{i_{n}})$, so $\gamma$ is a partial homomorphism.

$(2) \Rightarrow (1)$. Conversely, suppose Duplicator has a winning strategy in the $n$-round Ehrenfeucht-\Fraisse game (with $n \leq k$). For every possible set of Spoiler moves in the $n$-round game, i.e. for every sequence $s = [a_1,\dots,a_n]$, there exists a sequence $t = [b_1,\dots,b_n]$ such that $\gamma = \{(a_i,b_i) \mid 1 \leq i \leq n \}$ is a partial homomorphism. Let $\gamma_{s}$ be the partial homomorphism resulting from Duplicator's responses to Spoiler's play $s$. We define $f:\Ek\As \rightarrow \Bs$ by $f(s_i) = \gamma_s(a_i)$ where $s_i \sqsubseteq s$ ending in $a_i$. Since $\gamma_s$ is a partial homomorphism for all $s \in \Ek\As$, $f$ is a homomorphism. Namely, suppose $s_1,\dots,s_m $ are such that $R^{\Ek\As}(s_1,\dots,s_m)$. By the pairwise comparability condition, there exists some $\sqsubseteq$-greatest $s = [a_1,\dots,a_{l}]$ with $l \leq n$ amongst the $s_i$. Since $s_i \sqsubseteq s$, the last move of $s_i$, $\ve_{A}(s_i)$ must be in $\{a_1,\dots,a_{l}\}$. By the compatiblity condition, we have that $R^{\As}(\ve_{A}(s_1),\dots,\ve_{A}(s_m))$. Hence, since $\gamma_{s}$ is a partial homomorphism, $R^{\Bs}(\gamma_{s}(\epsA(s_1)),\dots,\gamma_{s}(\epsA(s_m)))$. Since $f(s_i) = \gamma_{s}(\epsA(s_i))$, we have that $R^{\Bs}(f(s_1),\dots,f(s_m))$.  
\end{proof}

\subsection{The Pebbling Comonad}

We now turn to the case of pebble games. The following construction appeared in \cite{abramsky2017pebbling}.
Given a structure $\As$, we define $\Tk \As$, which will represent plays of the $k$-pebble game on $\As$.\footnote{In \cite{abramsky2017pebbling} we used the notation $\mathbb{T}_k$ for this comonad.} The universe is  $(\kset \times A)^{+}$, the set of finite non-empty sequences of moves $(p, a)$, where $p \in \kset$ is a pebble index, and $a \in A$. We shall use the notation $s = [ (p_1, a_1), \ldots , (p_n, a_n)]$ for these sequences, which may be of arbitrary length. Thus the universe of $\Tk \As$ is always infinite, even if $\As$ is a finite structure. This is unavoidable, by 
\cite[Theorem 7]{abramsky2017pebbling}, which shows that no finite construction is equivalent to $\Tk$, even at the level of the preorder collapse of the coKleisli category. We  define $\epsA : \Tk A \rarr A$ to send a play  $[ (p_1, a_1), \ldots , (p_n, a_n)]$ to $a_n$, the $A$-component of its last move.

Given an $n$-ary relation $R \in \sg$, we define $R^{\Tk \As}(s_1, \ldots , s_n)$ iff (1) the $s_i$ are pairwise comparable in the prefix ordering;
(2) the pebble index of the last move in each $s_i$ does not appear in the suffix of $s_i$ in $s_j$ for any $s_j \prefgt s_i$; and (3) $R^{\As}(\epsA(s_1), \ldots , \epsA(s_n))$.

Note that this differs from the definition of $\REk$ just in condition (2). Intuitively, Spoiler moves by placing one from a fixed set of pebbles on an element of $A$; Duplicator responds by placing their matching pebble on an element of the other structure.
Thus there is a ``sliding window'' on the structures, of fixed size. It is this size which bounds the resource, not the length of the play. In particular, placing a pebble on an element implies removing it from the element on which it was previously placed. Condition (2) enforces that only the \emph{current} position of a pebble is taken into account when deciding whether elements are related.

Finally, given a homomorphism $f : \Tk \As \rarr \Bs$, we define $f^* : \Tk A \rarr \Tk B$ by \\
$f^*  [ (p_1, a_1), \ldots , (p_j, a_j)] = [(p_1, b_1), \ldots , (p_j, b_j)]$,
where $b_i = f [(p_1, a_1), \ldots , (p_i, a_i)]$, $1 \leq i \leq j$.

\begin{proposition}\label{pebbcomonprop}
The triple $(\Tk, \varepsilon, (\cdot)^*)$ is a comonad in Kleisli form.
\end{proposition}

This is \cite[Theorem 4]{abramsky2017pebbling}. The following is  \cite[Theorem 13]{abramsky2017pebbling}.
\begin{theorem}
\label{strmorth}
The following are equivalent:
\begin{enumerate}
\item There is a homomorphism $\Tk \As \rarr \Bs$.
\item There is a winning strategy for Duplicator in the existential $k$-pebble game from $\As$ to $\Bs$.
\end{enumerate}
\end{theorem}

\subsection{The Modal Comonad}
 
For the modal case, we assume that the relational vocabulary $\sg$ contains only symbols of arity at most 2. We can thus regard a $\sigma$-structure as a Kripke structure for a multi-modal logic, where the universe is thought of as a set of worlds, each binary relation symbol $\Ralph$ gives the accessibility relation for one of the modalities, and each unary relation symbol $P$ give the valuation for a corresponding propositional variable. If there are no unary symbols, such structures are exactly the labelled transition systems widely studied in concurrency \cite{milner1989communication}.

Modal logic localizes its notion of satisfaction in a structure to a world. We shall reflect this by using the category of \emph{pointed relational structures} $\CSp$. Objects of this category are pairs $(\As, a)$ where $\As$ is a $\sg$-structure and $a \in A$. Morphisms $h : (\As, a) \rarr (\Bs, b)$ are homomorphisms $h : \As \rarr \Bs$ such that $h(a) = b$. Of course, the same effect could be achieved by expanding the vocabulary $\sg$ with a constant, but pointed categories appear in many mathematical contexts.

For each $k \geq 0$, we shall define a comonad $\Mk$, where $\Mk (\As, a)$ corresponds to unravelling the structure $\As$, starting from $a$, to depth $k$. 
The universe of $\Mk (\As, a)$ comprises the unit sequence $[a]$, which is the distinguished element,  together with all sequences of the form $[a_0, \alpha_1, a_1, \ldots , \alpha_j, a_{j}]$, where $a = a_0$, $1 \leq j  \leq k$, and $\RA_{\alpha_i}(a_i, a_{i+1})$, $0 \leq i < j$. Note that the universe of $\Mz (\As, a)$ is $\{ [a] \}$.

The map $\epsA : \Mk (A, a) \rarr  (A, a)$ sends a sequence to its last element. Unary relation symbols $P$ are interpreted by $\PMA(s)$ iff $\PA(\epsA s)$. For binary relations $\Ralph$, the interpretation is $\RMA_{\alpha}(s,t)$ iff for some $a' \in A$, $t = s[\alpha, a']$.
 Given a homomorphism $f : \Mk (\As, a) \rarr (\Bs, b)$, we define $f^* : \Mk (\As, a) \rarr \Mk (\Bs, b)$ recursively by $f^*[a] = [b]$,
 $f^*(s[\alpha, a']) = f^*(s)[\alpha, b']$ where $b' = f(s[\alpha, a'])$.

\begin{proposition}
The triple $(\Mk, \varepsilon, (\cdot)^*)$ is a comonad in Kleisli form on $\CSp$.
\end{proposition}
\begin{proof}
This follows very similar lines to the proofs of Propositions~\ref{efcomonprop} and~\ref{pebbcomonprop}.
\end{proof}


\noindent
We recall the notion of \emph{simulation} between Kripke structures \cite{blackburn2002modal}. Given structures $\As$, $\Bs$, we define relations ${\simord_k} \; \subseteq \; A \times B$, $k \geq 0$, by induction on $k$: $a \simord_{k+1} b$ iff
        \begin{enumerate}
          \item for all unary $P$, $\PA(a)$ implies $\PB(b)$ \label{sim:1}
          \item for all $\Ralph$, if $\RaA(a, a')$, then for some $b'$, $\RaB(b, b')$ and $a' \simord_k b'$. \label{sim:2}
        \end{enumerate}
The base case $a \simord_{0} b$ holds whenever the first condition is satisfied. It is standard that these relations are equivalently formulated in terms of a modified existential Ehrenfeucht-\Fraisse game \cite{blackburn2002modal,gradel2014freedoms}.

\begin{theorem}
\label{simthm}
Let $\As$, $\Bs$ be Kripke structures, with $a \in A$ and $b \in B$, and $k \geq 0$.
The following are equivalent:
\begin{enumerate}
\item There is a homomorphism $f : \Mk (\As, a) \rarr (\Bs, b)$.
\item $a \simord_k b$.
\item There is a winning strategy for Duplicator in the $k$-round simulation game from $(\As, a)$ to $(\Bs, b)$.
\end{enumerate}
\end{theorem}
\begin{proof}
(1) $\Rightarrow$ (2) We prove this by induction on $k$. For the base case $k = 0$, by (1), there exists a homomorphism, $f:\mathbb{M}_{0}(\As,a) \rarr (\Bs,b)$. Hence, for all unary $P$, $\PA(a) \Rightarrow P^{\mathbb{M}_1(\As,a)}([a]) \Rightarrow \PB(f[a]) \Rightarrow \PB(b)$, so $a \simord_{0} b$. For the inductive step, let $k' = k+1$ and assume there is a homomorphism $f:\Mkp(\As,a) \rightarrow (\Bs,b)$. Suppose that $\RaA(a,a')$, then there exists a $b' = f[a,\alpha,a']$ such that $\RaB(b,b')$. Consider the homomorphism $f_{\alpha,a'}:\Mk(\As,a') \rightarrow (\Bs,b')$ given by $f[a'] = b'$ and $f_{\alpha,a'}(s) = f([a,\alpha,a']s)$. Applying the inductive hypothesis, we obtain that $a' \simord_k b'$, verifying condition 2 for a simulation. To verify condition 1, by the definition of $\Mkp(\As,a)$ and $f$ being a homomorphism we have:  $\PA(a) \Rightarrow P^{\Mkp(\As,a)}([a]) \Rightarrow \PB(f[a]) \Rightarrow \PB(b)$. Therefore, $a \simord_{k+1} b$ as desired. 

(2) $\Rightarrow$ (3) This is a standard result, and is detailed in section 2.7 of \cite{blackburn2002modal}.

(3) $\Rightarrow$ (1) We construct $f_{i}$ for $i \leq k$ by induction. For the base case $k = 0$, define $f_{0}:\mathbb{M}_{0}(\As,a) \rarr (\Bs,b)$ as $f_{0}[a] = b$. 
This is trivially a homomorphism.
Suppose, by the inductive hypothesis, $f_{k}:\Mk(\As,a) \rarr (\Bs,b)$ has been constructed. Let $k' = k+1$ and we define $f_{k'}:\Mkp(\As,a) \rarr (\Bs,b)$ by splitting into two cases. For the first case, consider $t \in \Mkp(\As,a)$ such that $t = s[\alpha,a']$ for some $s \in \Mk(\As,a)$, $\Ralph \in \sigma$ and $\RaA(a,a')$. We can take $s$ as the record of the Spoiler's moves in the first $k$ rounds of the simulation game. In this case, we let $f_{k'}(t) = b'$ where $b'$ is Duplicator's response to Spoiler choosing the $\alpha$-transition to $a'$ in the $k'$-th round. For the second case, consider $t \in \Mk(\As,a) \subseteq \Mkp(\As,a)$. In this case, define $f_{k'}(t) = f_k(t)$. We need to verify $f_{k'}$ is indeed a homomorphism. Suppose $s,t \in \Mkp(\As,a)$ are such that $R_{\alpha}^{\Mkp (\As, a)}(s,t)$. If $s,t \in \Mk(\As,a) \subseteq \Mkp(\As,a)$, then $f_{k'}(s) = f_{k}(s)$ and $f_{k'}(t) = f_{k}(t)$. Therefore, by the inductive hypothesis that $f_{k}$ is a homomorphism, $\RaB(f_{k'}(s),f_{k'}(t))$,
and moreover, for each unary predicate $P$, $P^{\Mk(\As, a)}(t)$ implies $P^{(\Bs, b)}(f_{k'}(t))$. By the definition of $R_{\alpha}^{\Mkp(\As,a)}$, the only other case is $t = s[\alpha,a']$. Since $f_{k'}(t)$ is Duplicator's response to the $k'$-round in her winning strategy, we can conclude that $\RaB(f_{k'}(s),f_{k'}(t))$, and
for each unary predicate $P$, $P^{\Mk(\As, a)}(t)$ implies $P^{(\Bs, b)}(f_{k'}(t))$.
\end{proof}

\section{$I$-morphisms}
\label{Imorsec}

Before turning to logical equivalences, there is a technical issue which needs to be addressed.

A coKleisli morphism $f : \Ek \As \to \Bs$ is an \emph{$I$-morphism} if $s \preford t$ and $\epsA(s) = \epsA(t)$ implies that $f(s) = f(t)$. 
\begin{proposition}
Consider $\sgp = \sigma \cup \{I\}$, where $I$ is a binary relation symbol not in $\sg$. If we interpret $I^{\As}$ and $I^{\Bs}$ as the identity relations on $A$ and $B$, then $f : \Ek \As \to \Bs$ is an $I$-morphism iff it is a $\sgp$-homomorphism.
\end{proposition}
\begin{proof}
Suppose $f$ is an $I$-morphism, and $I^{\Ek \As}(s,t)$.
By definition, this means that $s \comp t$  and $I^{\As}(\epsA(s), \epsA(t))$, \ie $\epsA(s) = \epsA(t)$. By the $I$-morphism definition, this implies $f(s) = f(t)$, and by the interpretation of $I^{\Bs}$ as the identity relation, $I^{\Bs}(f(s),f(t))$. Hence, $f$ preserves $I$, and is therefore a $\sgp$-homomorphism. The proof of the converse is similar.  
\end{proof}
The significance of the $I$-morphism condition becomes apparent from the following proposition.
\begin{proposition}
\label{Imorpisoprop}
If $f : \Ek \As \to \Bs$ and $g : \Ek \Bs \to \As$ are $I$-morphisms, then $f^*(s) =t$, $g^*(t) = s$ imply that $(s,t)$ defines a partial isomorphism from $\As$ to $\Bs$.
\end{proposition}
\begin{proof}
Let $s = [a_1,\dots,a_n]$ and $t = [b_1,\dots,b_n]$ with $n \leq k$, then the partial function defined by $(s,t)$ is $\gamma:a_i \mapsto b_i$. Suppose $a_j = a_l$ for some $j,l \leq n$ and let $s_j,s_l$ be the subsequences of $s$ ending in $a_j,a_l$ (respectively). Since $s_j \comp s_l$ and $\epsilon_{A}(s_j) = a_j = a_l = \epsilon_{A}(s_l)$, by the $I$-morphism condition, $f(s_j) = f(s_l)$. By definition of $f^*$ and $\gamma$, $b_j = b_l$. Hence, $\gamma$ is well-defined. Suppose $R \in \sigma$ is an $m$-ary relation symbol with $\RA(a_{i_1},\dots,a_{i_m})$ for $a_{i_j} \in s$ and let $s_{i_j}$ denote the prefix of $s$ such that $\epsilon_{A}(s_{i_j}) = a_{i_j}$. Since each $s_{i_j}$ is a prefix of $s$, the set of all $s_{i_j}$ is pairwise comparable, so $\REk(s_{i_1},\dots,s_{i_m})$ implying that $\RB(f(s_{i_1}),\dots,f(s_{i_m}))$ by $f$ being a homomorphism. This yields that $\RB(b_{i_1},\dots,b_{i_m})$, since $f(s_{i_j}) = b_{i_j}$ follows from $f^*(s) = t$. Hence, $\gamma$ preserves relations, and therefore is a partial homomorphism. Similarly, define $\gamma^{-1}:b_i \mapsto a_i$, which is also well-defined and a partial homomorphism. By construction, $\gamma$ and $\gamma^{-1}$ are inverses, so $(s,t)$ defines a partial isomorphism from $\As$ to $\Bs$.
\end{proof}

In the light of this result, the reader may wonder why we do not simply work with a sub-comonad of $\Ek$, in which only non-repeating sequences are allowed. The problem with this is that the functorial action of $\Ek$ on non-injective maps will produce repeating sequences. The same problem arises with the coKleisli extension.

A similar notion of $I$-morphism applies to the pebbling comonad \cite{abramsky2017pebbling}. 
Note that, if $I^{\As}$ is the identity relation on $\As$, $I^{\Pk \As}(s, t)$ holds if $s = s'[(p,a)]$, $t = t'[(p', a)]$, $s$ and $t$ are comparable in the prefix order, say $s \preford t$, and the placing of $p$ in $s$ is still current in $t$. If $p \neq p'$, this means that two different pebbles are placed on the same element of $\As$. We call such a sequence \emph{duplicating}. Duplicating sequences are to $\Pk$ what repeating sequences are to $\Ek$.

Dan Marsden and Tom\'a\v{s} Jakl have shown\footnote{Personal communication} that it is possible to define a retract of the $\Ek$ comonad in a more complicated fashion, which does eliminate repeating sequences. However, it is not clear how this can be extended to the pebbling comonad.

We shall develop an alternative, general approach, based on \emph{relative comonads}
\cite{altenkirch2010monads}. 
We shall not in fact need the full generality of relative comonads. What we shall use is the following special case, mentioned in \cite{altenkirch2010monads}. Given a functor $J : \CC \to \mathcal{D}$, and a comonad $G$ on $\mathcal{D}$, we  obtain a relative comonad on $\CC$ whose coKleisli category is defined as follows.
A morphism from $A$ to $B$, for objects $A$, $B$ of $\CC$, is a $\mathcal{D}$-arrow $G J A \to J B$. The counit at $A$ is $\ve_{J A}$, using the counit of $G$ at $J A$. Given $f : G J A \to J B$, the coKleisli extension $f^* : G J A \to G J B$ is the coKleisli extension of $G$.
Since $G$ is a comonad, these operations satisfy the equations for a comonad in Kleisli form.
It is in fact simpler to describe relative comonads in coKleisli form, as explained in \cite{altenkirch2010monads}.

To apply this to our setting,
we can take advantage of the fact that our comonads are defined uniformly in the relational vocabulary.
Given a relational vocabulary $\sg$, there is a full and faithful embedding $J : \CS \to \CSI$ such that $I^{J \As}$ is the identity on $A$.
Moreover, we have a comonad $\Ek^{I}$, which is our $\Ek$ construction applied to $\CSI$.
We use this data to obtain the $J$-relative comonad $\Ekp$ on $\CS$ as defined above.
Explicitly, we have the following coKleisli category: objects are those of $\CS$, morphisms from $\As$ to $\Bs$ are $\CS$-morphisms $\Ek^{I} J \As \to J \Bs$.
The counit at $\As$ is $\ve_{J\As} : \Ek^{I} J \As \to J \As$, the counit for $\Ek^{I}$ at $J \As$. The coKleisli extension of $f : \Ek^{I} J \As \to J \Bs$ is $f^* : \Ek^{I} J \As \to \Ek^{I} J \Bs$, defined using the coKleisli extension for $\Ek^{I}$. 

The important point is that the morphisms in this $J$-relative coKleisli category are exactly the $I$-morphisms in the coKleisli category for the comonad $\Ek$ on $\CS$.

This construction is perfectly general.
In particular, it can be carried out in exactly the same fashion for the pebbling comonad. In this case, it is the \emph{duplicating} sequences which are problematic, \ie those in which different  pebbles are placed on the same element \cite{abramsky2017pebbling}.
This condition is discussed in more detail in Section~\ref{logeqIIsec}.
The relative comonad approach provides a uniform solution.

\section{Logical Equivalences I}

We now show how our game comonads can be used to give syntax-free characterizations of a range of logical equivalences, which play a central r\^ole in finite model theory and modal logic.

Two equivalences can be defined uniformly for any indexed family of comonads $\Ck$:
\begin{itemize}
\item $\As \eqaCk \Bs$ iff there are coKleisli morphisms $\Ck \As \rarr \Bs$ and $\Ck \Bs \rarr \As$. Note that there need be no relationship between these morphisms.
This is simply the equivalence induced by the preorder collapse of the coKleisli category.
\item $\As \eqcCk \Bs$ iff $\As$ and $\Bs$ are isomorphic in the coKleisli category $\Kl(\Ck)$. This means that there are morphisms $\Ck \As \rarr \Bs$ and $\Ck \Bs \rarr \As$ which are inverses of each other in $\Kl(\Ck)$.
\end{itemize}
Clearly, $\eqcCk$ strictly implies $\eqaCk$.

We introduced a number of logic fragments in Section~\ref{fragssec}.
Each of these logics $\LL$ induces an equivalence on structures in $\CS$:
\[ \As \eqL \Bs \iffdef \forall \vphi \in \LL. \; \As \models \vphi \; \IFF \; \Bs \models \vphi . \]
Our aim is to characterize these equivalences in terms of our game comonads, and more specifically, to use morphisms in the coKleisli categories as witnesses for these equivalences. 

We shall use the two equivalences $\eqaCk$ and $\eqcCk$ induced by comonads as described above to characterize the logical equivalences induced by the existential-positive and counting versions of these logic fragments, $\ELk$ and $\Lck$, respectively. We shall use the relative versions $\Ekp$ and $\Pkp$ of the Ehrenfeucht-\Fraisse and pebbling comonads, as described in the previous section. These impose the condition of being $I$-morphisms, which will be important to get a precise match with the winning conditions of the corresponding logical games.

We shall defer the treatment of back-and-forth equivalences to Section~\ref{backandforthsec}, as it will build on material to be developed later. 

We now turn to a detailed study of each of our comonads in turn.

\subsection{The Ehrenfeucht-\Fraisse comonad}
\label{logeqsecEF}
\label{sec}

We begin by using the $I$-morphism condition enforced by $\Ekp$ to sharpen Theorem~\ref{EFgamethm} from an equivalence to a bijective correspondence.
\begin{theorem}
\label{EFgameeqthm}
There is a bijective correspondence between:
\begin{enumerate}
\item $\Ekp$ coKleisli morphisms from $\As$ to $\Bs$.
\item Winning strategies for Duplicator for the existential Ehrenfeucht-\Fraisse game with $k$ rounds, played from $\As$ to $\Bs$.
\end{enumerate}
\end{theorem}
\begin{proof}
The proof proceeds as for Theorem~\ref{EFgamethm}, but we now use the fact that the $I$-morphism condition ensures that if $f^*(s) = t$, the correspondence $\epsA s_i \mapsto \epsB t_i$ is single-valued, and hence satisfies the partial homomorphism winning condition for the existential game. Indeed, if $\epsA s_i = \epsA s_j$, then $I^{\Ek \As}(s_i,s_j)$.
Since $f$ is an $I$-morphism, this implies $I^{\Ek \Bs}(t_i,t_j)$, and hence $\ve_{B} t_i = \ve_{B} t_j$.
Conversely, a winning strategy will
produce pairs $(s,t)$ which satisfy the partial homomorphism condition, and in particular single-valuedness. Hence it induces a coKleisli map satisfying the $I$-morphism condition.
\end{proof}

We now recall the \emph{bijection game} \cite{Hella1996}. In this variant of the Ehrenfeucht-\Fraisse game, Spoiler wins if the two structures have different cardinality. Otherwise, at round $i$, Duplicator chooses a bijection $\psi_i$ between $A$ and $B$, and Spoiler chooses an element $a_i$ of $A$. This determines the choice by Duplicator of $b_i = \psi_i(a_i)$. Duplicator wins after $k$ rounds if the relation $\{ (a_i, b_i) \mid 1 \leq i \leq k \}$ is a partial isomorphism.

\begin{proposition}
\label{bijgameprop}
The following are equivalent, for finite structures $\As$ and $\Bs$:
\begin{enumerate}
\item $\As \eqcEk \Bs$.
\item There is a winning strategy for Duplicator in the $k$-round bijection game.
\end{enumerate}
\end{proposition}
\begin{proof}
Assuming (1), we have $I$-morphisms $f : \Ek \As \to \Bs$ and $g : \Ek \Bs \to \As$ with $g^* = {f^*}^{-1}$. 
For each $s \in \{ []\} \cup A^{<k}$, where $[]$ is the empty sequence, and $A^{<k}$ is the set of sequences over $A$ of length strictly less than $k$, we can define a map $\psi_s : A \to B$, by $\psi_s (a) = f(s[a])$. This is a bijection, with inverse $\chi_{t}$, $t = f^*(s)$, defined similarly from $g$. These bijections provide a strategy for Duplicator. Since by Proposition~\ref{Imorpisoprop}, each $(s, f^*(s))$ is a partial isomorphism, this is a winning strategy.

Conversely, a winning strategy provides bijections $\psi_s$, which we can use to define $f$ by $f(s[a]) = \psi_s(a)$. The winning conditions imply that this is an $I$-isomorphism in the coKleisli category.
\end{proof}


We can now state our main result on logical equivalences for the Ehrenfeucht-\Fraisse comonad.
\begin{theorem}
\begin{enumerate}
\item For all structures $\As$ and $\Bs$: $\As  \eqELk \Bs \; \IFF \; \As \eqaEk \Bs$.
\item For all finite structures $\As$ and $\Bs$: $\As \eqLck \Bs \; \IFF \; \As \eqcEk \Bs$.
\end{enumerate}
\label{thm:mainEF}
\end{theorem}
\begin{proof}
(1) In \cite{kolaitis1990expressive}, it is shown that $\As  \eqELk \Bs$ iff Duplicator has a winning strategy in the $k$-round existential Ehrenfeucht-\Fraisse game from $\As$ to $\Bs$.
Combining this with Theorem~\ref{EFgamethm} yields the result. \\
(2) In \cite{Hella1996} (see also the exposition in \cite[Theorem 8.13]{Libkin2004})  it is shown that $\As \eqLck \Bs$ iff Duplicator has a winning strategy in the $k$-round bijection game.
Combining this with Proposition~\ref{bijgameprop} yields the result.
\end{proof}


\subsection{The Pebbling Comonad}



We now state the following result, characterizing the equivalences induced by finite-variable logics $\Lvk$.
\begin{theorem}
\begin{enumerate}
\item For all structures $\As$ and $\Bs$: $\As  \eqELvk \Bs \; \IFF \; \As \eqaPk \Bs$.
\item For all finite structures $\As$ and $\Bs$: $\As \eqLvck \Bs \; \IFF \; \As \eqcPk \Bs$.
\end{enumerate}
\label{thm:mainPebble}
\end{theorem}
\begin{proof}
This follows from Theorems 14 and 18  of \cite{abramsky2017pebbling}.
\end{proof}


\subsection{The Modal Comonad}

Corresponding to the graded modal logic described in Section~\ref{fragssec}, a notion of \emph{graded bisimulation} is given in  \cite{Rijke2000}. This is in turn related to \emph{resource bisimulation} \cite{corradini1999graded}, which has been introduced in the concurrency setting. The two notions are shown to coincide for image-finite Kripke structures in \cite{aceto2010resource}, where it is also shown that they can be presented in a simplified form.
We recall that a Kripke structure $\As$ is image-finite if for all $a \in A$ and $\Ralph$, $\Ralph^{\As}(a) \, := \, \{ a' \mid \RA(a,a') \}$  is finite.

Adapting the results in \cite{aceto2010resource}, we define approximants $\gb_k$ for graded bisimulation, by induction on $k \geq 0$. We define  $a \gb_{k+1} b$ to hold iff: 
  \begin{enumerate}
    \item for all $P$, $\PA(a)$ iff $\PB(b)$
    \item for all $\Ralph$, there is a bijection $\theta : \RA_{\alpha}(a) \cong \RB_{\alpha}(b)$ such that, for all $a' \in \RA_{\alpha}(a)$, $a' \gb_k \theta(a')$.
  \end{enumerate}
The base case, $a \gb_{0} b$, holds whenever only the first condition is satisfied. We can also define a corresponding graded bisimulation game between $(\As, a)$ and $(\Bs, b)$. 
\begin{itemize}
\item At round 0, the elements $a_0 = a$ and $b_0 = b$ are chosen. Duplicator wins if for all $P$, $\PA(a)$ iff $\PB(b)$, otherwise Spoiler wins.
\item  At round $i+1$, Spoiler chooses some $\Ralph$, and Duplicator chooses a bijection $\theta_i : \RaA(a_i) \cong \RaB(b_i)$. If there is no such bijection, Spoiler wins. Otherwise, Spoiler then chooses $a_{i+1} \in \RA(a_i)$, and $b_{i+1} := \theta_i(a_{i+1})$. Duplicator wins this round if for all $P$, $\PA(a_{i+1})$ iff $\PB(b_{i+1})$, otherwise Spoiler wins.
\end{itemize}
This game is evidently analogous to the bijection game we encountered previously.
\begin{proposition}
\label{gradedbisimprop}
The following are equivalent:
\begin{enumerate}
\item There is a winning strategy for Duplicator in the $k$-round graded bisimulation game between $(\As, a)$ and $(\Bs, b)$.
\item $a \gb_k b$.
\item $(\As, a) \eqcMk (\Bs, b)$.
\end{enumerate}
\end{proposition}
\begin{proof}
(1) $\Rightarrow$ (2) We prove the statement by induction on $k$. For the base case $k = 0$, by (1), Duplicator has a winning strategy for the $0$-round game, so for all $P$, $\PA(a)$ iff $\PB(b)$, hence $a \gb_0 b$. Let $k' = k+1$ and assume that Duplicator has a winning strategy in the $k'$-round graded bisimulation game between $(\As,a)$ and $(\Bs,b)$. By the winning condition of this strategy, for all $P$, $\PA(a)$ iff $\PB(b)$. Moreover, for all $\Ralph$, Duplicator responds to Spoiler choosing $\Ralph$ in the first round of the game with a bijection $\theta_{1}:\RaA(a) \cong \RaB(b)$. It follows that for all $a' \in \RaA(a)$, Duplicator has a winning strategy in the $k$-round game from $(\As,a')$ to $(\Bs,\theta_{1}(a'))$. This fact together with the inductive hypothesis shows that $a' \gb_k \theta_1(a')$. Hence, $a \gb_{k'} b$. 

(2) $\Rightarrow$ (3) We prove the statement by induction on $k$. For the base case $k = 0$, we construct the coextension morphisms $f^{*}:\mathbb{M}_{0}(\As,a) \rarr \mathbb{M}_{0}(\Bs,b)$ by $f^{*}[a] = [b]$ and $g^{*}:\mathbb{M}_{0}(\Bs,b) \rarr \mathbb{M}_{0}(\As,a)$ by $g^{*}[b] = [a]$. By (2), $a \gb_0 b$, so $f$ and $g$ preserve unary relations $P$. Binary relations $\Ralph$ are vacuously preserved since $\mathbb{M}_{0}(\As,a)$ and $\mathbb{M}_{0}(\Bs,b)$ only have one element with no self-transitions. Using these definitions, $f^{*} \circ g^{*}[b] = [b]$ and $g^{*} \circ f^{*}[a] = [a]$. Therefore, $f$ and $g$ are inverses in the coKleisli category. Let $k' = k+1$ and assume that $a \gb_{k'} b$. We need to define coextension morphisms $f^{*}:\Mkp(\As,a) \rarr \Mkp(\Bs,b)$ and $g^{*}:\Mkp(\Bs,b) \rarr \Mkp(\As,a)$. Since $a \gb_{k'} b$, for every $\Ralph$, there exists a bijection $\theta:\RaA(a) \cong \RaB(b)$ such that for every $a' \in \RaA(a)$, $a' \gb_k \theta(a')$. For every $a' \in \RaA(a)$, let $b' = \theta(a')$. We can conclude $(\As, a') \eqcMk (\Bs, b')$ from the inductive hypothesis and $a' \gb_k b'$. Hence, there exists a coKleisli isomorphism given by the pair $f^{*}_{a',\alpha}:\Mk(\As,a') \rarr \Mk(\Bs,b')$ and  $g^{*}_{b',\alpha}:\Mk(\Bs,b') \rarr \Mk(\As,a')$. For every $\Ralph$ and $a' \in \RaA(a)$, we define $f^{*}$ as $f^{*}([a,\alpha]s) = [b,\alpha]f^{*}_{a',\alpha}(s)$ (by $s \in \Mk(\As,a')$, $a'$ is the first element of $s$). Similarly, for $\theta(a') = b'$ we define $g^{*}$ as $g^{*}([b,\alpha]t) = [a,\alpha]g^{*}_{b',\alpha}(t)$. Finally, let $f^{*}[a] = [b]$ and $g^{*}[b] = [a]$. To verify that $f^{*}$ is a homomorphism, suppose $s \in \Mkp(\As,a)$. We have two cases, either $s = [a,\alpha]t, s' = [a,\alpha]t'$ for some $t,t' \in \Mk(\As,a')$ and $a' \in \RaA(a)$ or $s = [a]$ and $s' = [a,\alpha,a']$. Since the interpretations of these relations only depend on the end of the sequence, suppose $\PMpA([a,\alpha]s)$ and $\RMpA([a,\alpha]s,[a,\alpha]s')$ for $s,s' \in \Mk(\As,a')$, then we have the following chains of equivalences resulting from $f^{*}_{a',\alpha}$ being a homomorphism:
\begin{align*}
\PMpA([a,\alpha]s) &\Leftrightarrow \PMAp(s) \\
&\Leftrightarrow \PMBp(f^{*}_{a',\alpha}(s)) \\ 
&\Leftrightarrow \PMpB([b,\alpha]f^{*}_{a',\alpha}(s)) \\ 
&\Leftrightarrow \PMpB(f^{*}([a,\alpha]s)) \\ 
\RMpA([a,\alpha]s,[a,\alpha]s') &\Leftrightarrow \RMAp(s,s')  \\
&\Leftrightarrow \RMBp(f^{*}_{a',\alpha}(s),f^{*}_{a',\alpha}(s'))  \\
&\Leftrightarrow \RMpB(f^{*}([a,\alpha]s),f^{*}([a,\alpha]s'))
\end{align*}
 If $s = [a]$, then by the first condition of $a \gb_{k'} b$ and the definitions of $\PMpA$,$\PMpB$, we obtain the following chain of equivalences: $\PMpA([a]) \Leftrightarrow \PA(a) \Leftrightarrow \PB(b) \Leftrightarrow \PMpB([b]) \Leftrightarrow \PMpB(f^{*}([a]))$. If $s = [a]$ and $s' = [a,\alpha,a']$, then $\RaA(a,a')$. Therefore, $a' \in \RaA(a)$ and $b' = \theta(a')$ where $\theta$ is the bijection $\RaA(a) \cong \RaB(b)$. Hence, $\RaB(b,b')$ and $\RMpB([b],[b,\alpha,b'])$. A similar argument holds for verifying that $g^{*}$ is a homomorphism. We must show that $f^{*}$ and $g^{*}$ are inverses. To show that $g^{*}(f^{*}(s)) = s$ we consider two cases:
\begin{enumerate}
\item Suppose $s = [a]$, then $g^{*} \circ f^{*}[a] = g^{*}[b] = [a]$ 
\item Suppose $s = [a,\alpha]s'$ for some $\Ralph$, $a' \in \RaA(a)$, and $\theta(a') = b'$:
\begin{align*}
g^{*} \circ f^{*}([a,\alpha]s') &= g^{*}([b,\alpha]f_{a',\alpha}^{*}(s')) \\
&= [a,\alpha](g^{*}_{b',\alpha}(f_{a',\alpha}^{*}(s')) \\
&= [a,\alpha](\id_{\Mk(\As,a')}(s')) \\
&= [a,\alpha]s'
\end{align*}
\end{enumerate}
Hence, $g^{*} \circ f^{*} = \id_{\Mkp(\As,a)}$. The proof that $f^{*} \circ g^{*} = \id_{\Mkp(\Bs,b)}$ is similar.  Therefore, $f$ is a coKleisli isomorphism, so $(\As, a) \eqcMkp (\Bs, b)$.

(3) $\Rightarrow$ (1) Assume $f:\Mk(\As,a) \rarr (\Bs,b)$ is a coKleisli isomorphism.  Let $a_0 = a$, $s_0 = [a]$ and $b_0 = b$. The fact that $\PA(a_0) \Leftrightarrow \PB(b_0)$ follows from $f$ being a coKleisli isomorphism. At round $i+1$, if Spoiler chooses $\Ralph$, Duplicator chooses the bijection $\theta_{s_i,\alpha}:\RaA(a_i) \cong \RaB(b_i)$ defined by $\theta_{s_i,\alpha}(a') = f(s_i[\alpha,a'])$. Spoiler then chooses an $a_{i+1} \in \RaA(a_i)$; define $b_{i+1} = \theta_{s_i,\alpha}(a_{i+1})$ and $s_{i+1} = s_{i} [\alpha,a_{i+1}]$. We must show that this determines a winning move for Duplicator, i.e. for all unary $P \in \sigma$, $\PA(a_{i+1}) \Leftrightarrow \PB(b_{i+1})$. By $f$ being a coKleisli isomorphism and the definition of $\PMA$, we obtain the following chain of equivalences $\PA(a_{i+1}) \Leftrightarrow \PMA(s_i[\alpha,a_{i+1}]) \Leftrightarrow \PB(f(s_i[\alpha,a_{i+1}])) \Leftrightarrow \PB(\theta_{s_i,\alpha}(a_{i+1})) \Leftrightarrow \PB(b_{i+1})$.
\end{proof}

\begin{theorem}
\begin{enumerate}
\item For all Kripke structures $(\As, a)$ and $(\Bs, b)$: 
\[ (\As, a)  \eqEMk (\Bs, b) \;\; \IFF \;\; (\As, a) \eqaMk (\Bs,b). \]
\item For all image-finite Kripke structures $(\As,a)$ and $(\Bs,b)$: 
\[ (\As,a) \eqMck (\Bs,b) \;\; \IFF \;\; (\As,a) \eqcMk (\Bs,b). \]
\end{enumerate}
\end{theorem}
\begin{proof}
(1) It is a standard result (see e.g.~\cite[Theorem 2.78]{blackburn2002modal}) that $(\As,a)  \eqEMk (\Bs,b)$ iff $a \simord_k b$ and $b \simord_k a$.
Combining this with Proposition~\ref{simthm} yields the result. \\
(2) It is shown in \cite{Rijke2000} and \cite[Proposition 4.11]{aceto2010resource} that $(\As,a) \eqMck (\Bs,b)$ iff $a \gb_k b$.
Combining this with Proposition~\ref{gradedbisimprop} yields the result.
\end{proof}

\section{Coalgebras and combinatorial parameters}
\label{coalgsec}
Another fundamental aspect of comonads is that they have an associated notion of \emph{coalgebra}. A coalgebra for a comonad $(G, \varepsilon, \delta)$ is a morphism $\alpha : A \to G A$ such that the following diagrams commute:
\begin{center}
\begin{tikzcd}
A  \ar[r, "\alpha"] \ar[d, "\alpha"']
& G A \ar[d,  "\delta_{A}"] \\
G A  \ar[r, "G \alpha"] 
& G^2 A
\end{tikzcd}  
$\qquad \qquad$
\begin{tikzcd}
A \ar[r, "\alpha"] \ar[rd, "\id_A"']
& G A \ar[d, "\epsA"] \\
& A
\end{tikzcd}
\end{center}
Given $G$-coalgebras $\alpha : A \to GA$ and $\beta : B \to GB$, a coalgebra morphism from $\alpha$ to $\beta$ is a morphism $h : A \to B$ such that the following diagram commutes:
\begin{center}
\begin{tikzcd}
A \ar[r, "\alpha"] \ar[d, "h"']
& GA \ar[d, "Gh"] \\
B \ar[r, "\beta"'] & GB
\end{tikzcd}
\end{center}
This gives a category of coalgebras and coalgebra morphisms, denoted by $\EM(G)$, the \textit{Eilenberg-Moore category} of $G$.

Our use of indexed comonads $\Ck$ opens up a new kind of question for coalgebras. Given a structure $\As$, we can ask: what is the least value of $k$ such that a $\Ck$-coalgebra exists on $\As$?  We call this the \emph{coalgebra number} of $\As$. We shall find that for each of our comonads, the coalgebra number is a significant combinatorial parameter of the structure.

\subsection{The Ehrenfeucht-\Fraisse comonad and tree-depth}

A graph is $G = (V, {\adj})$, where $V$ is the set of vertices, and $\adj$ is the adjacency relation, which is symmetric and irreflexive.
A \emph{forest cover} for $G$ is a forest $(F, {\leq})$ such that $V \subseteq F$, and if $v \adj v'$, then $v \comp v'$.
The \emph{tree-depth} $\td(G)$ is defined to be $\min_{F} \hgt(F)$, where $F$ ranges over forest covers of $G$.\footnote{We formulate this notion in order-theoretic rather than graph-theoretic language, but it is equivalent to the definition in \cite{nevsetvril2006tree}.} It is clear that we can restrict  to forest covers of the form $(V, {\leq})$, since given a forest cover $(F, {\leq})$ of $G = (V, {\adj})$, $(V, \, {\leq} \cap V^2)$ is also a forest cover of $G$, and $\hgt(V) \leq \hgt(F)$. Henceforth, by forest covers of $G$ we shall mean those with universe $V$.

Given a $\sg$-structure $\As$, the \emph{Gaifman graph} $\Gf(\As)$ is $(A, \adj)$, where $a \adj a'$ iff for some relation $R \in \sg$, for some $(a_1, \ldots , a_n) \in \RA$, $a = a_i$, $a' = a_j$, $a \neq a'$. The \emph{tree-depth of $\As$} is $\td(\Gf(\As))$.

\begin{theorem}
\label{tdth}
Let $\As$ be a finite $\sg$-structure, and $k>0$. There is a bijective correspondence between
\begin{enumerate}
\item $\Ek$-coalgebras $\alpha : \As \rarr \Ek \As$.
\item Forest covers of $\Gf(\As)$ of height $\leq k$.
\end{enumerate}
\end{theorem}
\begin{proof}
        Suppose that $\alpha : \As \to \Ek \As$ is a coalgebra. For $a \in A$, let $\alpha(a) = [a_1, \ldots , a_j]$, then the action of the comultiplication $\delta_{\As}$ on $\alpha(a)$ is $[[a_1],[a_1,a_2],\dots,[a_1,\ldots,a_j]]$. Hence, the first coalgebra equation states that $\alpha(a_i) = [a_1, \ldots , a_i]$, $1 \leq i \leq j$. The second states that $a_j = a$. Thus $\alpha : A \to \Alk$ is an injective map whose image is a prefix-closed subset of $\Alk$. Defining $a \leq a'$ iff $\alpha(a) \preford \alpha(a')$ yields a forest order on $A$, of height $\leq k$. If $a \adj a'$ in $\Gf(\As)$, then for some $a_1, \ldots , a_n$ with $a = a_i$, $a' = a_j$, we have $\RA(a_1, \ldots , a_n)$. Since $\alpha$ is a homomorphism, we must have $R^{\Ek \As}(\alpha(a_1), \ldots , \alpha(a_n))$. By the pairwise comparability condition in the definition of $R^{\Ek \As}$, $\alpha(a_i) \comp \alpha(a_j)$, and so $a_i \comp a_j$. Thus $(A, {\leq})$ is a forest cover of $\Gf(\As)$, of height $\leq k$.

Conversely, given such a forest cover $(A, {\leq})$, for each $a \in A$, its predecessors form a chain $a_1 < \cdots < a_j$, with $a_j = a$, and $j \leq k$.
We define $\alpha(a) = [a_1, \ldots , a_j]$, which yields a map $\alpha : A \to \Alk$, which evidently satisfies the coalgebra equations. If $\RA(a_1, \ldots , a_n)$, then since $(A, {\leq})$ is a forest cover, we must have $a_i \comp a_j$ for all $i, j$, and hence $\alpha(a_i) \comp \alpha(a_j)$.
Thus $R^{\Ek \As}(\alpha(a_1), \ldots , \alpha(a_n))$, and $\alpha$ is a homomorphism.
\end{proof}

\noindent We write $\cnE(\As)$ for the coalgebra number of $\As$ with respect to the Ehrenfeucht-\Fraisse comonad.

\begin{theorem}
For all finite structures $\As$: $\td(\As) \, = \, \cnE(\As)$.
\end{theorem}
\begin{proof}
By Theorem~\ref{tdth}, for all  $k>0$, $\td(\As) \leq k$ iff $\cnE(\As) \leq k$.
\end{proof}

\subsection{The pebbling comonad and tree-width}

We review the notions of tree decompositions and tree-width. A tree $(T, {\leq})$ is a forest with a least element (the root). A tree is easily seen to be a meet-semilattice: every pair of elements $x, x'$ has a greatest lower bound $x \wedge x'$ (the greatest common ancestor). The path from $x$ to $x'$ is  the set
$\pth(x, x') := [x \wedge x', x] \cup [x \wedge x', x']$, where we use interval notation: $[y, y'] := \{ z \in T \mid y \leq z \leq y' \}$. 

A \textit{tree-decomposition} of a graph $G = (V, {\adj})$ is a tree $(T, {\leq})$ together with a labelling function $\lbfn : T \rarr \pow(V)$ satisfying the following conditions: 
\begin{itemize}
\item (TD1) for all $v \in V$, for some $x \in T$, $v \in \lbfn(x)$; 
\item (TD2) if $v \adj v'$, then for some $x \in T$, $\{ v, v' \} \subseteq \lbfn(x)$; 
\item (TD3) if $v \in \lbfn(x) \cap \lbfn(x')$, then for all $y \in \pth(x, x')$, $v \in \lbfn(y)$. 
\end{itemize}
The \emph{width} of a tree decomposition is given by $\max_{x \in T} |\lbfn(x)| -1$. We define the \textit{tree-width} $\tw(G)$ of a graph $G$ as $\min_{T} \mathsf{width}(T)$, where $T$ ranges over tree decompositions of $G$.
\begin{remark}\label{ordernicerem}
 We say that a tree decomposition is \emph{orderly} if it has the following property: for all $x \in T$, there is at most one $v \in \lambda(x)$ such that for all $x' < x$, $v \not\in \lambda(x')$.
By  \cite[Definition~13.1.4]{kloks1994treewidth}, \emph{nice} tree decompositions are orderly. Moreover, by  \cite[Lemma 13.1.2]{kloks1994treewidth}, the existence of a tree decomposition implies the existence of a nice tree decomposition of the same width. Hence, the existence of a tree decomposition implies the existence of an orderly one of the same width.
This observation will prove useful for our next result.
\end{remark}

We shall now give an alternative formulation of tree-width which will provide a useful bridge to the coalgebraic characterization. It is also interesting in its own right: it clarifies the relationship between tree-width and tree-depth, and shows how pebbling arises naturally in connection with tree-width.

A \textit{$k$-pebble forest cover} for a graph $G = (V, {\adj})$ is a forest cover $(V, {\leq})$ together with a pebbling function $p : V \to \kset$ such that, if $v \adj v'$ with $v \leq v'$, then for all $w \in (v,v']$, $p(v) \neq p(w)$.

The following result is implicit in \cite{abramsky2017pebbling}, but it seems worthwhile to set it out more clearly.

\begin{theorem}
Let $G$ be a finite graph. The following are equivalent:
\begin{enumerate}
\item $G$ has a tree decomposition of width $< k$.
\item $G$ has a $k$-pebble forest cover.
\end{enumerate}
\end{theorem}
\begin{proof}
$(1) \Rightarrow (2)$. Assume that $G = (V, \adj)$ has a tree decomposition $(T, {\leq}, \lambda)$ of width $< k$. 
As explained in Remark~\ref{ordernicerem},  without loss of generality we can assume that the given tree decomposition is orderly.

For any $v \in V$, the set of $x \in T$ such that $v \in \lambda(x)$ is non-empty by (TD1), and closed under meets by (TD3). Since $T$ is a tree, this implies that this set has a least element $\tau(v)$. This defines a function $\tau : V \to T$. The fact that the tree decomposition is orderly implies that $\tau$ is injective. We can define an order on $V$ by $v \leq v'$ iff $\tau(v) \leq \tau(v')$. This is isomorphic to a sub-poset of $T$, and hence is a forest order.

We define $p : V \to \kset$ by induction on this order. Assuming $p(v')$ is defined for all $v' < v$, we consider $\tau(v)$. Since the tree decomposition is orderly, this means in particular that $p(v')$ is defined for all $v' \in S  :=  \lambda(\tau(v)) \setminus \{v\}$. Since the decomposition is of width $< k$, we must have $|S| < k$. We set $p(v) := \min (\kset \setminus \{ p(v') \mid v' \in S \})$.

To verify that $(V, {\leq})$ is a forest cover, suppose that $v \adj v'$. By (TD2), for some $x \in T$, $\{ v, v' \} \subseteq \lambda(x)$. We have $\tau(v) \leq x \geq \tau(v')$, and since $T$ is a tree, we must have $\tau(v) \, \comp \, \tau(v')$, whence $v \, \comp \, v'$.

Finally, we must verify the condition on the pebbling function $p$. Suppose that $v \adj v'$, and $v < w \leq v'$. 
Since $v \adj v'$, for some $x$, $\{ v, v' \} \subseteq \lambda(x)$. But then $\tau(v) < \tau(w) \leq \tau(v') \leq x$. Since $v \in \lambda(\tau(v)) \cap \lambda(x)$, by (TD3), $v \in \lambda(\tau(w))$. By construction of the pebbling function, this implies $p(v) \neq p(w)$.

$(2) \Rightarrow (1)$. Suppose that $(V, {\leq}, p)$ is a $k$-pebble forest cover of $G$. We define a tree $T = V_{\bot}$ by adjoining a least element $\bot$  to $V$. We say that $v'$ is an active predecessor of $v$ if $v' \leq v$, and for all $w \in (v', v]$, $p(v') \neq p(w)$. We define the labelling function by setting $\lambda(v)$ to be the set of active predecessors of $v$; $\lambda(\bot) := \es$. Since $p |_{\lambda(v)}$ is injective, $|\lambda(v)| \leq k$.

We verify the tree decomposition conditions. (TD1) holds, since $v \in \lambda(v)$. (TD2) If $v \adj v'$, then $v \comp v'$. Suppose $v \leq v'$.
Then $v$ is an active predecessor of $v'$, and $\{ v, v' \} \subseteq \lambda(v')$.
(TD3) Suppose $v \in \lambda(v_1) \cap \lambda(v_2)$. Then $v$ is an active predecessor of both $v_1$ and $v_2$. This implies that for all $w \in \pth(v_1,v_2)$, $v$ is an active predecessor of $w$, and hence $v \in \lambda(w)$.
\end{proof}

\begin{theorem}
\label{kpfcth}
Let $\As$ be a finite $\sg$-structure. There is a bijective correspondence between:
\begin{enumerate}
\item $\Pk$-coalgebras $\alpha : \As \to \Pk \As$
\item $k$-pebble forest covers of $\Gf(\As)$.
\end{enumerate}
\end{theorem}
\begin{proof}
See \cite[Theorem 6]{abramsky2017pebbling}.
\end{proof}
\noindent We write $\cnP(\As)$ for the coalgebra number of $\As$ with respect to the pebbling comonad.

\begin{theorem}
For all finite structures $\As$: $\tw(\As) \, = \, \cnP(\As) - 1$.
\end{theorem}

\subsection{The modal comonad and synchronization tree height}
Let $\As$ be a Kripke structure. It will be convenient to write labelled transitions $a \labar{\alpha} a'$ for $\Ralph(a, a')$.
Given $a \in A$, the submodel generated by $a$, denoted $\Sa$, is obtained by restricting the universe to the set $S_a$ of $a'$ such that there is a path $a \labar{\alpha_1} \cdots \labar{\alpha_k} a'$. This submodel forms a \emph{synchronization tree} \cite{milner1980calculus} if for all $a'$, there is a unique such path. The height of such a tree is the maximum length of any path from the root $a$.

\begin{proposition}\label{modalcoalgprop}
Let $\As$ be a Kripke structure, with $a \in A$. The following are equivalent:
\begin{enumerate}
\item There is a coalgebra $\gamma : (\Sa,a) \to \Mk (\Sa, a)$.
\item $\Sa$ is a synchronization tree of height $\leq k$.
\end{enumerate}
\label{prop:modalCoalg}
\end{proposition}
\begin{proof}
        Suppose there is a coalgebra $\gamma:(\Sa,a) \rarr \Mk(\Sa,a)$. For $a' \in S_a$, let $\gamma(a') = [a_0,\alpha_{1},a_{1},\dots,\alpha_{j},a_j]$ where $a_0 = a$. The first coalgebra equation states that $\gamma(a_{i}) = [a_0,\alpha_{1}\dots,\alpha_{i},a_{i}]$ for $0 \leq i \leq j$. The second coalgebra equation states that $a_j = a'$. Therefore, $\gamma$ is injective. By injectivity and since $
        \gamma$ is a homomorphism and thus preserves transitions, the path of transitions determined by $\gamma(a')$, i.e. $a \labar{\alpha_{1}} a_1 \cdots \labar{\alpha_{j}} a'$ is unique. Hence, every $a' \in S_a$ has a unique path from $a$, so the submodel generated by $a$ is a synchronization tree. Since $j \leq k$, the height of the tree is at most $k$.

        Conversely, suppose $\Sa$ is a synchronization tree of height $\leq k$, then for every $a'$ in this submodel, there exists a unique path of transitions $a \labar{\alpha_{1}} \cdots \labar{\alpha_{j}} a'$. We can define the morphism $\gamma:(\Sa,a) \rarr \Mk(\Sa,a)$ as $\gamma(a') = [a,\alpha_{1},a_1,\dots,\alpha_{j},a']$. Since the path of transitions is unique, the first coalgebra equation is satisfied. Moreover, since the last element of $\gamma(a')$ is $a'$, the second coalgebra equation is satisfied.  
\end{proof}

\noindent We define the modal depth $\md(\As, a) = k$ if the submodel $\Sa$ generated by $a$ is a synchronization tree of height $k$.

\begin{theorem}
Let $\As$ be a Kripke structure, and $a \in A$ be such that the submodel generated by $a$ is a synchronization tree of finite height.
Then $\md(\As, a) = \cnM(\As, a)$.
\end{theorem}

Note the conditional nature of this result, which contrasts with those for the other comonads. The modal comonad is defined in such a way that the universe $\Mk (A, a)$ reflects information about the possible transitions. Thus having a coalgebra at all, regardless of the value of the resource parameter, is a strong constraint on the structure of the transition system.

\section{Characterization of Rossman-type equivalences}
\label{sec:Rossman}

As a simple application of the results of the previous section, we show how they yield characterizations of some approximation preorders on structures, and their associated equivalences.

We begin with a completely general result. Given objects $A$, $B$ of a category $\CC$, we shall use the notation $A \to B$ to mean that there exists a morphism from $A$ to $B$ in $\CC$.

\begin{proposition}
\label{coalgapproxprop}
Let $G$ be a comonad on a category $\CC$. For all objects $A$, $B$ of $\CC$, the following are equivalent: 
\begin{enumerate}
    \item $GA \to B$
    \item For all $G$-coalgebras $\alpha : C \to GC$, $C \to A \IMP C \to B$
\end{enumerate}
\end{proposition}
\begin{proof}
Suppose $f : GA \to B$, $\alpha : C \to GC$ is a coalgebra, and $g : C \to A$. Then
$f \circ Gg \circ \alpha : C \to B$.

Conversely, suppose that for all coalgebras $\alpha : C \to GC$, $C \to A$ implies $C \to B$. Then since $\delta_A : GA \to GGA$ is a coalgebra, and $\varepsilon_A : GA \to A$, we conclude that $GA \to B$.
\end{proof}

Rossman defined the following preorder in \cite{Rossman2008}: $\As \arEk \Bs$ iff for all structures $\Cs$ with $\td(\Cs) \leq k$, $\Cs \to \As \IMP \Cs \to \Bs$. The associated equivalence $\eqk$ plays a major role in  \cite{Rossman2008}, as a resource-bounded approximation to homomorphism equivalence.

We can similarly define $\As \arPk \Bs$ in terms of treewidth: $\As \arPk \Bs$ iff for all structures $\Cs$ with $\tw(\Cs) < k$, $\Cs \to \As \IMP \Cs \to \Bs$ \footnote{We use strict inequality, in this case, due to the fact that tree-width is one less than the $\Pk$-coalgebra number.}.

Finally, we can define $(\As, a) \arMk (\Bs, b)$ iff for all synchronization trees $(\Cs, c)$ with $\md(\Cs, c) \leq k$, $(\Cs, c) \to (\As, a) \IMP (\Cs, c) \to (\Bs, b)$.

As an immediate consequence of Proposition~\ref{coalgapproxprop} and the results of the previous section, we obtain:

\begin{proposition}
For all structures $A$, $B$ in $\CS$:
\begin{enumerate}
\item $\As \arEk \Bs \;\; \IFF \;\; \Ek \As \to \Bs$
\item $\As \arPk \Bs \;\; \IFF \;\; \Pk \As \to \Bs$
\item $(\As, a) \arMk (\Bs, b) \;\; \IFF \;\; \Mk (\As, a) \to (\Bs, b)$ (when $\As,\Bs$ are Kripke structures)
\end{enumerate}

\end{proposition}

\section{Coalgebras and conjunctive queries}
\label{sec:conjunctive}

We shall now connect coalgebras with logic. We will show that they are closely related to \emph{conjunctive queries}.
We shall focus on boolean conjunctive queries, \ie closed formulas built from atomic formulae using only conjunction and existential quantification. We allow the empty conjunction $\top$.
Given a finite structure $\As$ with universe $A = \{ a_1, \ldots , a_n\}$, the canonical conjunctive query for $\As$ is a formula $\QA := \exists v_1 , \ldots , v_n. \bigwedge \{ R(\vv) \mid R \in \sg, \va \in \RA\}$.
Here the correspondence $v_i \leftrightarrow a_i$ uses the  linear ordering on the universe implicitly given by the enumeration.

The key property of the canonical conjunctive query is the following \cite{chandra1977optimal}:
\begin{theorem}
For all finite structures $\As$, $\Bs$, $\As \to \Bs$ iff $\Bs \models \QA$.
\end{theorem}

We now consider a number of rewrite rules on formulas:
\[ \begin{array}{ll}
(R1) & \mbox{Associative-Commutative-Identity rewriting of conjunctions} \\
(R2) & \exists v.\, (\vphi \AND \psi) \rew (\exists v.\, \vphi) \AND \psi\;\;(v \not\in \FV(\psi)) \\
(R3) & \exists v. \, \exists w. \, \vphi \rew \exists w. \, \exists v. \, \vphi \\
(R4) & \exists v.\, \vphi \rew \exists w.\, \phi[w/v]
\end{array}
\]
Note that $(R1)$ amounts to taking formulas as a commutative monoid under the binary operation of conjunction, with $\top$ as the identity.

\subsection{The Ehrenfeucht-\Fraisse case}
As usual, we can use a rule $\theta \rew \chi$ to rewrite a formula $C[\theta]$ to $C[\chi]$. We write $\vphi \redE \psi$ if $\vphi$ can be rewritten to 
$\psi$ in some finite number of steps using rules (R1)--(R3).

\begin{theorem}
Let $\As$ be a finite structure, and $\QA$ its canonical conjunctive query. The following are equivalent:
\begin{enumerate}
\item $\As$ has a coalgebra $\As \to \Ek \As$
\item $\QA \redE \vphi$, where $\vphi$ has quantifier rank $\leq k$.
\end{enumerate}
\label{thm:efRewrite}
\end{theorem}
\begin{proof}
Firstly, we note that given any conjunctive query $\chi$ satisfied by $\As$, with variable $v_i$ interpreted by $a_i$, there is a corresponding forest order on $A$, with $a_i < a_j$ if $v_j$ occurs in the scope of $\exists v_i$ in $\chi$.

Now assuming (2), there is a  forest order on $A$ corresponding to $\vphi$.  Since the quantifier rank of $\vphi$ is $\leq k$, so is the height of this order.
Moreover, the order induced by $\QA$ is obviously a forest cover of the Gaifman graph $\Gf(A)$, and application of the rules (R1)--(R3) preserves that property of the orders induced by conjunctive queries.  In particular, in applying (R2), the free variable condition implies that no elements which become incomparable in the corresponding order are adjacent in the Gaifman graph. By induction on the length of the rewrite sequence, we conclude that the order induced by $\vphi$ is a forest cover, and hence, by Theorem~\ref{tdth}, determines an $\Ek$-coalgebra on $\As$.

For the converse, we argue by induction on the cardinality of $A$, or equivalently on the number $n$ of variables in $\QA$. 
We prove the following statement $\Phi(n)$ by induction on $n$:

For all finite relational vocabularies $\sigma$, and for all $\sigma$-structures $\As$ of cardinality $n$, if $\As$ has a forest cover of height $k$, then $\QA \redE \vphi$, where the quantifier rank of $\vphi$ is $\leq k$.

The base case $n=1$ is trivial. Now consider 
\[\QA = \exists v_1 , \ldots , v_j. \exists v_{j+1}, \ldots , v_n . \bigwedge_{i \in I} R_i(\vec{v}_i) . \]
We assume that the variables are in an order which linearizes the given forest order on $A$. In particular, $v_{j+1}, \ldots , v_n$ are the leaves (maximal elements) of maximum height $k$ in the forest order. Note that if $v$ and $v'$ are leaves, they cannot be adjacent in the Gaifman graph of $\As$, since this would violate the forest cover condition. Hence we can partition $I$ as $I = I_{j+1} \sqcup \cdots I_n \sqcup L$, where $I_p$ is the set of indices labelling atomic formulas in which $v_p$ occurs, $p = j+1, \ldots , n$. Using (R1) and (R2), we can rewrite $\QA$ to 
\[ \vphi \, := \, \exists v_1 , \ldots , v_j. \, (\bigwedge_{p=j+1}^n \chi_p) \wedge \theta \]
where $\chi_p := \exists v_p. \, \bigwedge_{j \in I_p} R_j(\vv_j)$ and $\theta := \bigwedge_{l \in L} R_l(\vv_l)$.
We now define 
\[ \vphi' := \exists v_1 , \ldots , v_j. (\bigwedge_{p=j+1}^n R_p(\ww)) \wedge \theta , \]
where $R_p$ is a new relation symbol, and $\ww_p$ lists all the variables other than $v_p$ occurring in $\chi_p$. 
This formula $\vphi'$ is the canonical conjunctive query for a new structure $\As'$, of cardinality $j<n$. There are many new adjacencies in the Gaifman graph of $\As'$ compared with that for $\As$. However, note that, since (the element of $A$ labelled by) each $v_p$ was maximal in the given forest cover of $\As$, the set of variables adjacent to $v_p$ must be linearly ordered in that forest cover. Hence the restriction of the forest cover on $\As$ to $\{ v_1, \ldots , v_j \}$ is a forest cover on $\As'$, of height $k-1$.
Our induction hypothesis can be applied to this forest cover on $\As'$ as a $\sigma'$-structure, where $\sigma' = \sigma \cup \{R_p \mid p = j+1, \ldots , n \}$, yielding a rewrite $\vphi' \redE \psi'$, where $\psi'$ has quantifier rank $ m < k$.
We can perform the same rewrite steps to obtain $\vphi \redE \psi$, 
where $\psi$ results from $\psi'$ by replacing each $R_p(\ww)$ by $\chi_p$. We have $\QA \redE \vphi \redE \psi$. Moreover, the quantifier rank of $\psi$ is $m+1 \leq k$, as required.
\end{proof}
We can think of (R1)--(R3) as encoding a non-deterministic algorithm for computing the minimum quantifier rank for $\QA$ (and hence the tree-depth of $\As$). We use (R3) to guess an order on the variables, and then apply (R1) and (R2) to the quantifier prefix from the inside out. The above argument shows that this algorithm does, for some choice of order, find the minimum quantifier rank.

\subsection{The pebbling case}
We can give a very similar analysis for the pebbling case. We define $\vphi \redP \psi$ if $\vphi$ can be rewritten to 
$\psi$ in some finite number of steps using rules (R1), (R2), and (R4). This set of rules is widely used in query optimization \cite{dalmau2002constraint,ullman1989database}.
\begin{theorem}
Let $\As$ be a finite structure, and $\QA$ its canonical conjunctive query. The following are equivalent:
\begin{enumerate}
\item $\As$ has a coalgebra $\As \to \Pk \As$
\item $\QA \redP \vphi$, where $\vphi$ has number of variables $\leq k$.
\end{enumerate}
\end{theorem}
\begin{proof}
This is the statement of Theorem 27 in \cite{abramsky2017pebbling}.
\end{proof}

\subsection{The modal case}
For the modal case, we move to the category of pointed modal structures $\CSp$. 
Given a pointed modal structure $(\As, a)$, there is a natural notion of \emph{canonical modal conjunctive query} from $a$, defined recursively as follows:
\[ \MQAa \; \equiv \; \bigwedge \{ p \mid \As, a \models p \} \wedge \bigwedge \{ \Diamond_{\alpha} \MQAb \mid \Ralph(a,b) \} 
\]
Note that, even if $A$ is finite, this definition may not be well-founded if there are transition cycles. However, if the submodel $\Sa$ generated by $a$ is finite and acyclic, then this definition yields a well-defined finite modal formula.
In particular, this holds if $\Sa$ is a finite synchronization tree.\footnote{Note that if $\Sa$ is finite and acyclic, it is bisimilar to a finite synchronization tree \cite[Proposition 2.15]{blackburn2002modal}.}
This may be regarded as the ``existential-conjunctive half'' of a modal form of Scott sentence for the structure \cite{scott2014logic}.
Under the modal translation, this formula will map into a first-order formula $\psi(x)$ in one free variable, whose quantifier rank matches the modal depth of the formula. 

To relate this to the rewrite rules,
we  modify the standard canonical conjunctive query $\QS$ to have one free variable; we write this as $\QAp(x)$. Note that we can rearrange the block of existential quantifiers in $\QS$ using rule (R3) so that $\exists x$ is the leftmost existential quantifier. Hence, we could write $\QS$ as $\exists x \QAp(x)$, so $(\Sa,a) \vDash \QAp(x) \IFF \Sa \vDash \QS$. 

We write $\varphi(x) \redM \psi(x)$ if $\varphi(x)$ can be rewritten to $\psi(x)$ in some finite number of steps using rules (R1)--(R4).

\begin{theorem}
Let $(\As, a)$ be a finite modal structure. 
If $\Sa$ has a coalgebra 
$(\Sa,a) \rarr \Mk(\Sa,a)$,
then 
$\QAp(x) \redM \lsem \MQAa \rsem$, where $\MQAa$ is the canonical modal conjunctive query for $(\As, a)$, with modal depth $k$.
\end{theorem}
\begin{proof}
By Proposition \ref{prop:modalCoalg}, it suffices to show that: for every $a \in \As$, if $\Sa$ is a synchronization tree of height $\leq k$, then $\QAp(x) \redM \phi(x)$ where $\phi(x) = \lsem \MQAa \rsem$ and $\MQAa$ has modal depth $\leq k$. We prove this statement by induction on $k$. For the base case $k = 0$, we can conclude that $\Sa$ is such that $\Ralph^{\Sa}$ is empty for every $\Ralph \in \sg$. Therefore, $\QAp(x)$ is a conjunction of unary predicates $P \in \sigma$ that $a$ satisfies. We can take $\phi(x) = \QAp(x)$ and $\MQAa$ would be a conjunction of the propositional variables corresponding to the unary relations satisfied by $a$. For the inductive step, suppose $\Sa$ is a synchronization tree of height $\leq k+1$, then we can successively apply (R1) and (R2), to obtain a conjunction of formulas of two forms: 
\begin{enumerate}
\item $\exists y (R_{\beta}(x,y) \wedge \QBp(y))$, for some $y$ corresponding to $b$ such that $\RbA(a,b)$.
\item $P(x)$ for some unary $P \in \sg$ such that $\PA(a)$.
\end{enumerate}
Applying the inductive hypothesis to $\Sb$, and replacing $\QBp(y)$ with its rewrite $\chi(y) = \lsem \MQAb \rsem$ of modal depth $\leq k$ yields the desired $\phi(x)$.

\end{proof}

Note that, by our previous discussion, the converse to this result holds ``up to bisimulation''.
That is, if $(\As,a)$ has a canonical modal conjunctive query $\MQAa$ of modal depth $k$, then $\Sa$ is bisimilar to a synchronization tree $(T,a)$ of depth $k$. By Proposition~\ref{modalcoalgprop}, $T$ has a coalgebra 
$(T,a) \rarr \Mk(T,a)$.

\section{Characterization of Eilenberg-Moore categories}
\label{sec:adjoints}
Monads and comonads are closely related to \emph{adjunctions}. We recall that an adjunction
\[ \begin{tikzcd}
\CC \arrow[r, bend left=25, ""{name=U, below}, "L"]
\arrow[r, leftarrow, bend right=25, ""{name=D}, "R"{below}]
& \DD
\arrow[phantom, "\bot", from=U, to=D] 
\end{tikzcd}
\]
between categories $\CC$ and $\DD$ is given by functors $L : \CC \to \DD$ (the \emph{left adjoint}) and $R : \DD \to \CC$ (the \emph{right adjoint}), together with natural transformations $\eta_A : A \to RL A$ (the \emph{unit} of the adjunction), and $\ve_{B} : LR B \to B$ (the \emph{counit} of the adjunction), such that the maps
$\theta_{A,B} : \CC(A, RB) \to \DD(LA, B)$, and $\theta'_{A,B} : \DD(LA, B) \to \CC(A, RB)$, defined by
\[ \theta_{A,B}(f) = \ve_B \circ Lf, \qquad \theta'_{A,B}(g) = Rg \circ \eta_{A}, \]
are mutually inverse.

Each such adjunction gives rise to a monad on $\CC$, and a comonad on $\DD$. The comonad is $(G, \ve, \delta)$, where $G = LR$, and $\delta_B : LRB \to LRLRB$ is given by $\delta_B = L(\eta_{RB})$.

Conversely, every comonad arises from an adjunction, known as a \emph{resolution} of the comonad, in this way. In fact, there is a category of such resolutions for a given comonad $G$. The minimal (initial) resolution is the adjunction associated with the coKleisli category $\Kl(G)$, while the maximal (terminal) resolution arises from the category of coalgebras of $G$, also known as the \emph{Eilenberg-Moore category} $\EM(G)$ \cite{eilenberg1965adjoint}.

We have already studied the structure of the coalgebras for our game comonads. We shall now complete this analysis, by characterising the Eilenberg-Moore categories in terms independent of the comonads which give rise to them. This will give a new perspective on these constructions, as universal solutions to the problem of building various kinds of resource-bounded tree-structured covers of a given relational structure.

\subsection{The Ehrenfeucht-\Fraisse adjunction}

We define a \emph{tree-ordered $\sg$-structure} $(\As, {\leq})$ to be a $\sg$-structure $\As$ with a forest order $\leq$ on $A$, satisfying the following condition: 
\begin{center}
(E) if $a \adj b$ in $\Gf(\As)$, the Gaifman graph of $\As$ as a $\sg$-structure, then $a \comp b$.   
\end{center}
Note that such a structure is the same thing as a forest cover of $\Gf(\As)$ with universe $A$.
We write the covering relation of $\leq$ as $\cvr$.
A morphism of tree-ordered $\sg$-structures $f : (\As, {\leq}) \to (\Bs, {\leq'})$ is a $\sg$-homomorphism $f : \As \to \Bs$ which maps roots to roots, and preserves the covering relation. This determines a category $\RT(\sg)$. For each $k>0$, if we restrict to forest orders of height $\leq k$, we get a sub-category $\RTk(\sg)$.

\begin{theorem}
\label{EFadj}
For each $k >0$, the Eilenberg-Moore category $\EM(\Ek)$ is isomorphic to $\RTk(\sg)$.
\end{theorem}
\begin{proof}
By Theorem~\ref{tdth}, the objects of the two categories are in bijective correspondence.
Thus it remains to show that a $\sg$-homomorphism $f : \As \to \Bs$ is an $\Ek$-coalgebra morphism iff it is an 
$\RTk(\sg)$ morphism. From the analysis in Theorem~\ref{tdth}, if $\alpha : \As \to \Ek \As$ is a coalgebra, then
$\alpha(a) = [a_1, \ldots , a_k]$ iff $a_1 \cvr \cdots \cvr a_k = a$ is the covering chain of predecessors of $a$ in the forest order. Thus a coalgebra morphism preserves covering chains, and in particular preserves roots and the covering relation. Conversely, if a morphism preserves roots and the covering relation, then it preserves covering chains, and hence is a coalgebra morphism.
\end{proof}

As an immediate consequence of this result, we can describe the canonical resolution of $\Ek$ given by the adjunction between $\CS$ and $\EM(\Ek)$ in terms of $\RTk(\sg)$. 
There is an evident  forgetful functor $U_k : \RTk(\sg) \to \CS$, which simply forgets the forest order.

\begin{theorem}
The functor $U_k$ has a right adjoint given by $G_k$, where $G_k(\As) = (\Ek \As, {\preford})$.
The comonad arising from this adjunction is $\Ek$.
\end{theorem}

\subsection{The pebbling adjunction}
We define a \emph{$k$-pebble tree-ordered $\sg$-structure} $(\As, {\leq}, p)$ to be a $\sg$-structure $\As$ with a forest order $\leq$ on $A$, and a pebbling function $p: A \to \kset$. In addition to condition (E), it must also satisfy the following condition: 
\begin{center}
(P) if $a \adj b$ in $\Gsg(\As)$, and $a < b$ in the forest order, then for all $x \in (a, b]$, $p(a) \neq p(x)$.
\end{center}
Morphisms of these structures are morphisms of tree-ordered structures which additionally preserve the pebbling function. These define a category $\RPk(\sg)$ (where $k$ bounds  the number of pebbles, rather than the height of the forest order), and there is an evident forgetful functor $V_k : \RPk(\sg) \to \CS$.

\begin{theorem}
For each $k >0$, the Eilenberg-Moore category $\EM(\Pk)$ is isomorphic to $\RPk(\sg)$.
\end{theorem}
\begin{proof}
The proof proceeds on similar lines to that of Theorem~\ref{EFadj}. 
By Theorem~\ref{kpfcth}, the objects of the two categories are in bijective correspondence.
Thus to complete the argument, we must show that coalgebra morphisms coincide with $\RPk(\sg)$-morphisms.
Coalgebra morphisms preserve covering chains, and also preserve pebble indices, and this is equivalent to the conditions for $\RPk(\sg)$-morphisms.
\end{proof}

Once again, there is an immediate corollary to this result.
\begin{theorem}
The functor $V_k$ has a right adjoint $H_k$, where $H_k(\As) = (\Pk \As, {\preford}, p)$, where $p([(p_1, a_1), \ldots , (p_j, a_j)]) = p_j$. Moreover, $\Pk$ is the comonad arising from this adjunction.
\end{theorem}

\subsection{The modal adjunction}

Let $\sg$ be a modal vocabulary. We define a modal $\sg$-structure $(\As, a, {\leq})$ to be a $\sg$-structure $\As$ with a partial order $\leq$ on $A$, such that ${\uparrow}(a) := \{ a' \in A \mid a \leq a' \}$ is a tree order. This must satisfy the following condition: 
\begin{center} 
(M) for $x, y \in {\uparrow}(a)$,  $x \cvr y$ iff for some unique $\Ralph$, $\Ralph(x,y)$. 
\end{center}
A morphism of such structures is a $\sg$-homomorphism preserving the root and the covering relation. For each $k>0$, there is a category $\mathcal{R}^{M}_{k}(\sg)$, with $k$ bounding the height of the tree order, and a forgetful functor $W_k : \mathcal{R}^{M}_{k}(\sg) \to \CSp$, which sends $(\As, a, {\leq})$ to $(\As, a)$.

\begin{theorem}
For each $k >0$, the Eilenberg-Moore category $\EM(\Mk)$ is isomorphic to $\mathcal{R}^{M}_{k}(\sg)$.
\end{theorem}

\begin{theorem}
For each $k >0$, the functor $W_k$ has a right adjoint, which sends $(\As, a)$ to $(\Mk \As, [a], {\preford})$. Moreover, $\Mk$ is the comonad arising from this adjunction.
\end{theorem}

We have stated the result in this fashion for uniformity with the other cases, but it is worth noting that a simpler analysis applies in this case.\footnote{The authors are indebted to Dan Marsden for pointing this out.} The modal comonad is idempotent, \ie the comultiplication is an isomorphism.
This means that it encodes a coreflective subcategory.
The category of trees $\mathcal{R}^{M}_{k}(\sg)$ is coreflective in the category of all pointed Kripke structures (see e.g. \cite{winskel1995models}), and the modal comonad arises from this coreflection.

\subsection{$I$-morphisms and relative adjunctions}
\label{Imoreladjsec}

To fit $I$-morphisms, as discussed in section~\ref{Imorsec}, into this picture, we shall use \emph{relative adjunctions} \cite{ulmer1968properties}.
Given functors $J : \BB \to \DD$ and $L : \CC \to \DD$, we say that $L$ has a $J$-right adjoint $R : \BB \to \CC$ if there is a natural isomorphism
\[ \hom_{\DD}(L(-), J(-)) \; \cong \; \hom_{\CC}(-, R(-)) . \]
In particular, if $L$ has a \textit{bona fide} right adjoint $R : \DD \to \CC$, we obtain a $J$-right adjoint by composing $R$ with $J$.

In our case, we have the functor $J : \CS \to \CSI$, and we obtain relative versions of the Ehrenfeucht-\Fraisse and pebbling adjunctions, again using the fact that our constructions are uniform in the relational vocabulary, and hence can be applied to $\CSI$. These relative right adjoints give rise to the relative comonads $\Ekp$ and $\Pkp$.

\section{Logical Equivalences II: Open pathwise embeddings and back-and-forth equivalences}
\label{backandforthsec}

For each comonad $\Ck$, we shall now define a ``back-and-forth'' equivalence $\eqbCk$, intermediate between $\eqaCk$ and $\eqcCk$, and use it to characterize the logical equivalences induced by $\Lk$.
These back-and-forth equivalences will be more specific to ``game comonads'' defined on  $\CS$, but they will still be defined and shown to have the appropriate properties in considerable generality. We shall use a variant of the well-known notion of \emph{open map bisimulation} of Joyal, Nielsen and Winskel \cite{joyal1993bisimulation}. 

Although it is folklore that Ehrenfeucht-\Fraisse equivalence is ``essentially'' a form of bisimulation, to our knowledge this is the first time that graded elementary equivalences have been captured in a precise common format with open map bisimulations. The key novel ingredient in our approach is the use of embeddings, and the notion of \emph{pathwise embedding}.
This allows us to capture a general notion of \emph{property-preserving bisimulation}, which specializes to capture all the notions of interest for model comparison games.

Referring to our study of adjunctions characterizing the game comonads in Section~\ref{sec:adjoints},
we can extract the following common structure. For each comonad $\Ck$ (variously $\Ek$, $\Pk$, or $\Mk$), we have a category $\RCk$, and a commuting diagram of categories and functors:
\[ \begin{tikzcd}
& \RCk \arrow[dl] \arrow[dr, "U_k"] \\
\Trees \arrow[dr] & & \CS \arrow[dl] \\
& \Set &
\end{tikzcd}
\]
Here $\Trees$ is the category of tree orders and maps which preserve the root and the covering relation. The unlabelled arrows are the evident faithful forgetful functors. The faithful functor $U_k$ has a right adjoint $F_k$, and $\Ck$ arises from this adjunction.

Note that the Ehrenfeucht-\Fraisse and pebbling comonads produce \emph{forests} of non-empty sequences rather than trees. However, it is always possible, and often convenient, e.g.~for the base case of an inductive proof, to turn the forests into trees by adding the empty sequence. This is justified by the fact that $\Forests$, the category of forest orders and maps which take roots to roots, and preserve the covering relation, is equivalent to $\Trees$.

The above description applies directly to the Ehrenfeucht-\Fraisse and pebbling cases. For the modal case, we should use $\CSp$ rather than $\CS$.
In this case, the construction produces a tree directly.

We recall the notion of $I$-morphism from Section~\ref{Imorsec}.
We shall build this into the current picture by working with the $J$-relative versions of the Ehrenfeucht-\Fraisse and pebbling adjunctions, as in section~\ref{Imoreladjsec}. This means that objects of $\RCk$ carry an interpretation of the $I$-relation, which morphisms of $\RCk$ must preserve. The right adjoint functor $F_k$ produces a relation $I^{F_k \As}$ by the functorial action of $F_k$ on the identity relation $I^{\As}$.

A morphism $e$ in $\RCk$ is an \emph{embedding} if
$U_k(e)$ is an extremal mono in $\CS$, \ie an embedding of relational structures. We write  $e: \As \rightarrowtail \Bs$ to indicate that $e$ is an embedding.
 
We define a subcategory $\Paths$ of $\RCk$ whose objects have tree orders which are linear, so they comprise a single branch. If $P$ is a path, then $I^{P}$ is the identity relation. Morphisms of paths are embeddings. More generally, we say that $e : P \embed \As$ is a \emph{path embedding} if $P$ is a path. A morphism $f : \As \to \Bs$ in $\RCk$ is a \emph{pathwise embedding} if for any path embedding $e : P \embed \As$, $f \circ e$ is a path embedding.

We can now define what it means for a morphism $f : \As \to \Bs$ in $\RCk$ to be \emph{open}.
This holds if, whenever we have a diagram
\[ \begin{tikzcd}
  P \arrow[r, rightarrowtail] \arrow[d,rightarrowtail]
    & Q \arrow[d, rightarrowtail] \\
  \As \arrow[r,  "f"']
&\Bs 
\end{tikzcd}
\]
where $P$ and $Q$ are paths, there is an embedding $Q \rightarrowtail \As$ such that
\[ \begin{tikzcd}
  P \arrow[r, rightarrowtail] \arrow[d,rightarrowtail]
    & Q \arrow[dl, rightarrowtail] \arrow[d, rightarrowtail] \\
  \As \arrow[r,  "f"']
&\Bs 
\end{tikzcd}
\]
This is often referred to as the \emph{path-lifting property}. If we think of $f$ as witnessing a simulation of $\As$ by $\Bs$, path-lifting means that if we extend a given behaviour in $\Bs$ (expressed by extending the path $P$ to $Q$), then we can find a matching behaviour in $\As$ to ``cover'' this extension. Thus it expresses an abstract form of the notion of  ``p-morphism'' from modal logic \cite{blackburn2002modal}, or of functional bisimulation.

We can now define the \emph{back-and-forth equivalence} $\As \eqbCk \Bs$ between structures in $\CS$. This holds if there is a structure $\Rs$ in $\RCk$, and a span of open pathwise embeddings
\[ \begin{tikzcd}
& \Rs \arrow[dl] \arrow[dr] \\
F_k\As & & F_k \Bs
\end{tikzcd}
\]
This gives us a general, structural description of back-and-forth equivalence. At the same level of generality, we shall now define a \emph{back-and-forth game} $\Gk(\As,\Bs)$ played between the structures $\As$ and $\Bs$, corresponding to the comonad $\Ck$. Positions of the game are pairs $(s,t) \in \Ck A \times \Ck B$. 

We define a relation  $\WABC \, \subseteq \, \Ck A \times \Ck B$ as follows. A pair $(s,t)$ is in $\WABC$ iff for some path $P$, path embeddings $e_1 : P \embed \Ck \As$ and $e_2 : P \embed \Ck \Bs$, and $p \in P$, $s = e_1(p)$ and $t = e_2(p)$. The intention is that $\WABC$ picks out the winning positions for Duplicator.

At the start of each round of the game, the position is specified by $(s, t) \in \Ck A \times \Ck B$. The initial position is $(\bot, \bot)$.  The round proceeds as follows. Either Spoiler chooses some $s' \rcvr s$, and Duplicator must respond with $t' \rcvr t$, resulting in a new position $(s', t')$; or Spoiler chooses some $t'' \rcvr t$ and Duplicator must respond with $s'' \rcvr s$, resulting in $(s'',t'')$. Duplicator wins the round if they are able to respond, and  the new position is in $\WABC$.

\begin{theorem}
\label{pisoprop}
The following are equivalent:
\begin{enumerate}
\item $\As \eqbCk \Bs$.
\item There is a winning strategy for Duplicator in the $\Gk(\As,\Bs)$ game.
\end{enumerate}
\end{theorem}
\begin{proof}
$(1) \IMP (2)$. Firstly, consider an open pathwise embedding $q : \Rs \to F_k \Bs$. Since $q$ is a morphism in $\RCk$, any covering chain $\bot \cvr r_1 \cvr \cdots \cvr r_i$ in $\Rs$ is mapped to a covering chain $\bot \cvr s_1 \cvr \cdots \cvr s_i$ in $F_k \Bs$.
Since $\bot \cvr r_1 \cvr \cdots \cvr r_i$ is the image of a path embedding $P \embed \Rs$, and $q$ is a pathwise embedding, for each $j$ with $1 \leq j \leq i$, $(r_j,s_j) \in \WABC$.

\textbf{Claim} For any $s_{i+1} \rcvr s_i \in F_k B$, there is $r_{i+1} \rcvr r_i \in R$ such that $q(r_{i+1}) = s_{i+1}$.

To prove this, note that there are path embeddings $e_1 : P \embed \Rs$ and $e_2 : Q \embed F_k \Bs$ with images
$\bot \cvr r_1 \cvr \cdots \cvr r_i$ and $\bot \cvr s_1 \cvr \cdots \cvr s_i \cvr s_{i+1}$ respectively. Writing $Q_{\leq i}$ for the truncation of $Q$ to  the preimage of $\bot \cvr s_1 \cvr \cdots \cvr s_i$, $\Rs_i$ for the image of $e_1$, and $(F_k \Bs)_i$ for the image of $e_2 |_{Q_{\leq i}}$, we have an embedding $e : P \cong \Rs_i \cong (F_k \Bs)_i \cong Q_{\leq i}\embed Q$, such that $q \circ e_1 = e_2 \circ e$. Since $q$ is open, there is an embedding $Q \embed \Rs$ such that
\[ \begin{tikzcd}
  P \arrow[r, rightarrowtail, "e"] \arrow[d,rightarrowtail, "e_1"']
    & Q \arrow[dl, rightarrowtail] \arrow[d, rightarrowtail, "e_2"] \\
  \Rs \arrow[r,  "q"']
& F_k \Bs 
\end{tikzcd}
\]
We define $r_{i+1}$ to be the image of the greatest element of $Q$ under this embedding. Since $Q$ is a chain, $r_{i+1} \rcvr r_i$. By the commutativity of the above diagram, $q(r_{i+1}) = s_{i+1}$. 

Now suppose we have a span of open pathwise embeddings
\[ \begin{tikzcd}
& \Rs \arrow[dl, "p"'] \arrow[dr, "q"] \\
F_k\As & & F_k \Bs
\end{tikzcd}
\]
For any path embedding $P \embed \Rs$ with image $\bot \cvr r_1 \cvr \cdots \cvr r_i$, 
we have a pair $(s, t) = (p(r_i), q(r_i))$ in $\Ck \As \times \Ck \Bs$. Since $p$ and $q$ are pathwise embeddings, $(s, t) \in \WABC$.

Moreover, for any $t' \rcvr t$, we can find $r_{i+1} \rcvr r_i$ and $s' \rcvr s$ such that $(p(r_{i+1}), q(r_{i+1})) = (s',t')$ by the above Claim. A symmetric argument applies to yield an extension for any $s' \rcvr s$. Thus the set of such pairs $(s, t)$ yields a winning strategy for Duplicator in the $\Gk(\As, \Bs)$ game.

$(2) \IMP (1)$. Suppose we have a winning strategy for Duplicator. This determines a set $R  \subseteq \Ck A \times \Ck B$ of the plays following this strategy. 

$R$ is a down-closed subset of a tree, and hence forms a tree in the induced order. It determines a substructure of the product relational structure $F_k \As \times F_k \Bs$, and hence an object $\Rs$ in $\RCk$. Restricting the projections to $R$ determines morphisms $p : \Rs \to F_k \As$, $q : \Rs \to F_k \Bs$.

Firstly, we show that $p$ and $q$ are pathwise embeddings. Given a path embedding $e : P \embed \Rs$, $p \circ e$ is an injective homomorphism, since it is a morphism in $\RCk$, and preserves the covering relation. 
It remains to show that $p \circ e$, or equivalently the restriction of $p$ to the image of $e$, is strong.
Given $(s_1,t_1), \ldots , (s_n,t_n)$ in the image of $e$, let $s$ be the maximum of the $s_i$, which exists since the image of a path forms a chain. If $s = s_j$, then $t_j$ is the maximum of the $t_i$. Now since $(s,t) \in \WABC$, there is a path $Q$, and embeddings $e_1 : Q \embed F_k \As$, $e_2 : Q \embed F_k \Bs$, with $s = e_1(r)$ and $t = e_2(r)$ for some $r \in Q$.
Since $e_1$ and $e_2$ preserve covering chains, they map predecessors of $r$ bijectively to prefixes of $s$ and $t$ respectively, in a height-preserving fashion.\footnote{The height of an element of a forest is the length of its unique covering chain from a root.} Thus for each $i$, there is a unique
element $r_i \in Q$ such that $e_1(r_i) = s_i$ and $e_2(r_i) = t_i$. For an $n$-ary relation $R$ in $\sg$, since $e_1$ and $e_2$ are embeddings, we have 
\[ R^{F_k \As}(s_1, \ldots , s_n) \IFF R^{Q}(r_1, \ldots , r_n) \IFF R^{F_k \Bs}(t_1, \ldots , t_n) . \]
Hence 
\[ R^{F_k \As}(s_1, \ldots , s_n) \IFF (R^{F_k \As}(s_1, \ldots , s_n) \AND R^{F_k \Bs}(t_1, \ldots , t_n)) \IFF R^{\Rs}((s_1,t_1), \ldots , (s_n,t_n)), \]
and thus $p \circ e$ is a strong homomorphism, as required. The verification that $q$ is a pathwise embedding is entirely similar.

Finally, we must show that these morphisms are open.
Suppose we are given a diagram
\[ \begin{tikzcd}
  P \arrow[r, rightarrowtail] \arrow[d,rightarrowtail]
    & Q \arrow[d, rightarrowtail] \\
  \Rs \arrow[r,  "p"']
&F_k \As 
\end{tikzcd}
\]
We consider the case where $Q$ extends $P$ by one element; the general case will then follow by induction. 

Let the image of the maximum element of $P$ be $(s, t)$, and the image of the maximum element of $Q$ be $s' \rcvr s$. Since $R$ encodes the plays of a winning strategy for Duplicator, for some $t' \rcvr t$, we have $(s', t') \in R$, and moreover $(s', t') \in \WABC$. This implies that the map which extends the embedding $P \embed \Rs$ by sending the maximum element of $Q$ to $(s',t')$ determines an embedding $Q \embed \Rs$, which makes the path-lifting diagram commute.

The fact that $q$ is open is shown by a symmetric argument.
\end{proof}

For each of our three comonads, we will show that the game $\Gk$ specializes to exactly the expected concrete game: Ehrenfeucht-\Fraisse, pebbling, and modal bisimulation respectively.
The connections with the corresponding logical equivalences will follow as immediate corollaries.

\subsection{The Ehrenfeucht-\Fraisse comonad}

We now show that the generic game $\Gk$, when instantiated to the Ehrenfeucht-\Fraisse adjunction and comonad, yields exactly the standard Ehrenfeucht-\Fraisse game.

We begin with two observations we will use.
\begin{lemma}
\label{nrisolemm}
Consider $s = [a_1, \ldots , a_i] \in \Ek \As$.
If $s$ is non-repeating, then the substructure of $\Ek \As$ determined by the prefixes of $s$ is isomorphic to the substructure of $\As$ determined by $\{ a_1, \ldots , a_i \}$.
\end{lemma}
\begin{proof}
Since the elements of the first set form a chain, this is immediate from the definition of $\Ek \As$, given that the correspondence $[a_1, \ldots , a_j] \mapsto a_j$ is bijective, which holds if $s$ is non-repeating.
\end{proof}

\begin{lemma}
\label{nrembedlemm}
The image of a path embedding $e : P \embed F_k \As$ can only contain non-repeating sequences.
\end{lemma}
\begin{proof}
Since $e$ is an $I$-morphism, it must reflect the $I$ relation, while at the same time, it is injective. This implies that $I^{F_k \As}$ must be the identity relation on the image of $e$, or equivalently, that the image of $e$ must contain only non-repeating sequences.
\end{proof}

\begin{theorem}
\label{EFgeneric}
Given $\sg$-structures $\As$ and $\Bs$, the $\Gk(\As,\Bs)$ game for the Ehrenfeucht-\Fraisse comonad is equivalent to the Ehrenfeucht-\Fraisse game between $\As$ and $\Bs$.

\end{theorem}
\begin{proof}
Since $s' \rcvr s$ in $\Ek \As$ iff $s' = s[a]$ for some $a \in A$, we see from the definition of $\Gk(\As,\Bs)$ for the Ehrenfeucht-\Fraisse comonad
that it coincides with the Ehrenfeucht-\Fraisse game, provided we can show, for $(s,t) \in \Ek \As \times \Ek \Bs$, that $(s,t) \in \WABC$ iff $(s,t)$ satisfies the winning condition for the Ehrenfeucht-\Fraisse game.
We recall that this winning condition is that the relation $s_i \mapsto t_i$ is a partial isomorphism between $\As$ and $\Bs$.

Firstly, suppose that $(s,t) \in \WABC$, via path embeddings $e_1 : P \embed \Ek \As$, $e_2 : P \embed \Ek \Bs$.
This implies that the image of $e_1$, consisting of the prefixes of $s$, is isomorphic to the image of $e_2$, consisting of the prefixes of $t$. 
Applying Lemmas~\ref{nrisolemm} and \ref{nrembedlemm}, we conclude that $(s,t)$ determine a partial isomorphism between $\As$ and $\Bs$, as required.

For the converse, note that for any $(s, t)$ following a winning strategy for Duplicator in the Ehrenfeucht-\Fraisse game, the elementwise correspondence $s_i \mapsto t_i$ must be a partial isomorphism, and in particular a bijection. Hence Duplicator must always give the same response to repetitions of a given move by Spoiler. It follows that the behaviour of the strategy is determined by its restriction to non-repeating sequences, and without loss of generality, we can restrict the strategy to non-repeating sequences $(s,t)$. 
The prefixes of $s$ form an induced substructure of $\Ek \As$ which, applying lemma~\ref{nrisolemm}, is isomorphic to the substructure of $\Ek \Bs$ induced by the prefixes of $t$. Hence these substructures are the images of path embeddings with a common domain, and $(s,t) \in \WABC$.
\end{proof}

\begin{theorem}
\label{EFlogicequivthm}
For all structures $\As$ and $\Bs$: $\As \eqLk \Bs \; \IFF \; \As \eqbEk \Bs$.
\end{theorem}
\begin{proof}
This is an immediate corollary to Theorems~\ref{classicgamesthm}(1), \ref{pisoprop},  and~\ref{EFgeneric}.
\end{proof}

\begin{remark}
If we drop the $I$-morphism requirement on embeddings, we obtain a weaker notion of equivalence, in which the winning condition is that we have a \emph{partial correspondence} rather than a partial isomorphism. This yields 
a characterization of elementary equivalence for equality-free logic \cite{Casanovas1996}.
\end{remark}

\subsection{The pebbling comonad}
\label{logeqIIsec}

The analysis for the pebbling comonad is largely similar to that for the Ehrenfeucht-\Fraisse comonad.

Firstly, given $s = [(p_1, a_1), \ldots , (p_i, a_i)] \in \Pk \As$, and $p \in \{ p_1, \ldots , p_i \}$, we define
$\lastp(s) = a$, where $s = s_1[(p,a)]s_2$, and $p$ does not occur in $s_2$.
We say that $s$ is \emph{non-duplicating} if for all $s' \preford s$, $\lastp(s') = \lastpp(s') \IMP p = p'$. Thus we never have two pebbles placed on the same element.

Given $s = [(p_1, a_1), \ldots , (p_i, a_i)] \in \Pk \As$ and $t = [(p_1, b_1), \ldots , (p_i, b_i)] \in \Pk \Bs$, we define $\gamma(s,t) := \{ (\lastp(s), \lastp(t)) \mid p \in \{ p_1, \ldots , p_i \} \} \subseteq A \times B$. The prefix of $s$ (respectively $t$) of length $j$ is denoted $s_j$ (respectively $t_j$). The winning condition on plays $(s,t) \in \Pk A \times \Pk B$ for the pebbling game is that for each $j$, $\gamma(s_j,t_j)$ is a partial isomorphism between $\As$ and $\Bs$.

The following results play the role for the pebbling comonad of Lemmas~\ref{nrisolemm} and~\ref{nrembedlemm} for the Ehrenfeucht-\Fraisse comonad.


\begin{lemma}
\label{nonduplemm}
If $s$ and $t$ are non-duplicating, then the following are equivalent:
\begin{enumerate}
    \item For all $j$, $\gamma(s_j,t_j)$ is a partial isomorphism from $\As$ to $\Bs$.
   \item  The substructure $P_s$ of $\Pk \As$ determined by the prefixes of $s$ is isomorphic to the substructure $P_t$ of $\Pk \Bs$ determined by the prefixes of $t$.
\end{enumerate}
\end{lemma}
\begin{proof}
If $s$ and $t$ are non-duplicating, then each $\gamma(s_j,t_j)$ is a partial bijection.
Also, the correspondence $\beta : s_j \mapsto t_j$ is a bijection between $P_s$ and $P_t$.
For $s_1, \ldots , s_n \in P_s$, since they are all prefixes of $s$ and hence comparable,
$R^{\Pk \As}(s_1, \ldots , s_n)$ iff $R^{\As}(\epsA s_1, \ldots , \epsA s_n)$ and $\{ \epsA s_1, \ldots , \epsA s_n \} \subseteq \dom \, \gamma(s',t')$, where $s'$ is the maximum in the prefix ordering among $s_1, \ldots , s_n$. Similarly for relation instances in $P_t$.
Hence preservation of relation instances by $\beta$ is equivalent to preservation by $\gamma(s_j,t_j)$ for all $j$, and similarly for their inverses.
\end{proof}

\begin{lemma}
\label{ndupembedlemm}
The image of an embedding $e : P \embed G_k \As$ can only contain non-duplicating sequences.
\end{lemma}
\begin{proof}
Since $e$ is an $I$-morphism, it must reflect the $I$ relation, while at the same time, it is injective. This implies that $I^{G_k \As}$ must be the identity relation, or equivalently, that the image of $e$ must contain only non-duplicating sequences.
\end{proof}

\begin{theorem}
\label{Pgeneric}
Given $\sg$-structures $\As$ and $\Bs$, the $\Gk(\As,\Bs)$ game for the pebbling comonad is equivalent to the pebbling game between $\As$ and $\Bs$.
\end{theorem}
\begin{proof}
The proof proceeds essentially identically to that of Theorem~\ref{EFgeneric}, except that the appeal to 
Lemmas~\ref{nrisolemm} and~\ref{nrembedlemm} is replaced by one to Lemmas~\ref{nonduplemm} and~\ref{ndupembedlemm}.
\end{proof}

\begin{theorem}
\label{Plogicequivthm}
For all structures $\As$ and $\Bs$: $\As \eqLvk \Bs \; \IFF \; \As \eqbPk \Bs$.
\end{theorem}
\begin{proof}
This is an immediate corollary to Theorems~\ref{classicgamesthm}(2), \ref{pisoprop}, and~\ref{Pgeneric}.
\end{proof}

\subsection{The modal comonad}
The key notion of equivalence in modal logic is bisimulation \cite{blackburn2002modal,sangiorgi2009origins}. We shall define the finite approximants to bisimulation \cite{hennessy1980observing}.\footnote{Our focus on finite approximants in this paper is for uniformity, and because they are relevant in resource terms. We can extend the comonadic semantics beyond the finite levels. We shall return to this point in the final section.}
Given Kripke structures $\As$ and $\Bs$, we define a family of relations $\sim_k \; \subseteq \; A \times B$: $a \sim_{k+1} b$ iff 
\begin{enumerate}
    \item for all unary $P$, $\PA(a)$ iff $\PB(b)$
    \item for all binary $\Ralph$, $\RaA(a, a')$ implies for some $b'$, $\RaB(b,b')$ and $a' \sim_k b'$, and $\RaB(b,b')$ implies for some $a'$, $\RaA(a, a')$ and $a' \sim_k b'$.
\end{enumerate}
The base case $a \sim_{0} b$ holds whenever only the first condition is satisfied.

Plays in the \emph{modal bisimulation game} between pointed structures $(\As,a)$ and $(\Bs,b)$ are represented by positions $(s,t) \in \Mk (A, a) \times \Mk (B, b)$. 
A position $(s, t)$ satisfies the winning condition if $s = [a_0, \alpha_1, a_1, \ldots \alpha_i, a_i]$,
$t = [b_0, \alpha_1, b_1, \ldots \alpha_i, b_i]$, $a_0 = a$, $b_0 = b$, and for all $0 \leq j \leq i$, and unary predicates $P$, $\PA(a_j)$ iff $\PB(b_j)$. 
Note that there is no bijection requirement. This is why $I$-morphisms are not needed in the modal case.

The initial position is $([a], [b])$. If we have reached position $(s,t)$ after $k$ rounds, then round $k+1$ proceeds as follows:
Spoiler either chooses $s' \rcvr s$, and Duplicator must respond with $t' \rcvr t$, producing the new position $(s',t')$; or Spoiler chooses $t'' \rcvr t$, and Duplicator must respond with $s'' \rcvr s$, producing the new position $(s'',t'')$. Duplicator wins the round if they are able to respond, and the new position satisfies the winning condition.

Note that $s' \rcvr s$ iff $s' = s[\alpha, a']$ for a transition relation $\Ralph$ and $a' \in A$ such that  
$\RaA(a, a')$, where $a = \epsA(s)$.

The following standard result is essentially immediate from the definitions.
\begin{proposition}
Given pointed structures $(\As, a)$ and $(\Bs, b)$, the following are equivalent:
\begin{enumerate}
    \item $a \sim_k b$
    \item Duplicator has a winning strategy for the $k$-round modal bisimulation game.
\end{enumerate}
\end{proposition}

\begin{theorem}
\label{Mgeneric}
Given pointed structures $(\As, a)$ and $(\Bs, b)$, the $\Gk(\As,\Bs)$ game for the modal comonad is equivalent to the modal bisimulation game between $(\As, a)$ and $(\Bs, b)$.
\end{theorem}
\begin{proof}
Given the way we have defined the modal bisimulation game, this immediately reduces to verifying that the winning conditions coincide.
To see this, note that there is a path embedding $e : (R, r) \embed \Mk (\As, a)$ with $e(r_i) = s_i$ iff 
$r = r_0 \lt{\alpha_1} r_1 \lt{\alpha_2} \ldots , \lt{\alpha_i} r_i$, 
$s_i = [a_0, \alpha_1, a_1, \alpha_2, a_2, \ldots, \alpha_i, a_i]$, with $a = a_0$, and for all unary $P$, and
$0 \leq j \leq i$, $P^{Q}(r_i) \IFF P^{\As}(a_i)$.
Chaining these equivalences for a span of path embeddings shows
the equivalence of the winning conditions for $\Gk(\As,\Bs)$ and the modal bisimulation game.
\end{proof}

\begin{theorem}
\label{Mlogicequivthm}
For all pointed structures $(\As, a)$ and $(\Bs,b)$: $\As \eqMk \Bs \; \IFF \; \As \eqbMk \Bs$.
\end{theorem}
\begin{proof}
This is an immediate corollary to Theorems~\ref{classicgamesthm}(3), \ref{pisoprop}, and~\ref{Mgeneric}.
\end{proof}

\section{Further Directions}

From the categorical perspective, there is considerable additional structure which we have not needed for the results in this paper, but which may be useful for further investigations.

\textbf{Coequaliser requirements}
In Moggi's work on computational monads, there is an ``equaliser requirement'' \cite{moggi1991notions}. The dual version for a comonad $(G, \varepsilon, \delta)$ is that for every object $A$, the following diagram is a coequaliser:
\begin{center}
\begin{tikzcd}
G^2 A \ar[r, "G \epsA", yshift=0.9ex] \ar[r, "\varepsilon_{GA}"', yshift=-0.6ex]
& GA \ar[r, "\epsA"]
& A
\end{tikzcd}
\end{center}
This says in particular that the counit is a regular epi, and hence  $GA$ ``covers'' $A$ in a strong sense.

This coequaliser requirement holds for all our comonads. For $\Ek$, this is basically the observation that, given a sequence of sequences $[s_1, \ldots , s_j]$, we have $\ve[\ve s_1, \ldots , \ve s_j] \, = \, \ve s_j$. The other cases are similar.

\textbf{Indexed and graded structure}
Our comonads $\Ek$, $\Pk$, $\Mk$ are not merely discretely indexed by the resource parameter. In each case, there is a functor $(\Zp, {\leq}) \to \Comon(\CS)$
from the poset category of the positive integers to the category of comonads on $\CS$. Thus if $k \leq l$ there is a natural transformation with components $i^{k,l}_A : \Ek \As \to \El \As$, which preserves the counit and comultiplication; and similarly for the other comonads. Concretely, this is just including the plays of up to $k$ rounds in the plays of up to $l$ rounds, $k \leq l$.

Another way of parameterizing comonads by resource information is grading \cite{gaboardi2016combining}. Recall that comonads on $\CC$ are exactly the comonoids in the strict monoidal category $([\CC, \CC], {\circ}, I)$ of endofunctors on $\CC$ \cite{mac2013categories}. Generalizing this description, a graded comonad is an oplax monoidal functor $G : (M, {\cdot}, 1) \to ([\CC, \CC], {\circ}, I)$ from a monoid of grades into this endofunctor category. This means that for each $m \in M$, there is an endofunctor $G_m$, there is a graded counit natural transformation $\ve : G_{1} \Rightarrow I$, and for all $m, m' \in M$, there is a graded comultiplication $\delta^{m, m'} : G_{m\cdot m'} \Rightarrow G_m G_{m'}$.

The two notions can obviously be combined. We can see our comonads as (trivially) graded, by viewing them as oplax monoidal functors $(\Zp_{\infty}, {\leq}, \min, \infty) \to ([\CC, \CC], {\circ}, I)$. Given $k \leq l$, we have e.g. $\Ek \Rightarrow \Ek \Ek \Rightarrow \Ek \El$. The unit is interpreted using the $\omega$-colimit $\Eomega$, as described in the following paragraph.

The question is whether there are  interesting graded structures which arise naturally in considering richer logical and computational situations in our setting.

\textbf{Colimits and infinite behaviour}
In this paper, we have dealt exclusively with finite resource levels. However, there is an elegant means of passing to infinite levels. We shall illustrate this with the modal comonad. Using the inclusion morphisms described in the  previous discussion of indexed structure, for each pointed structure $(\As, a)$ we have a diagram
\[ \MM_1 (\As, a) \to \MM_2 (\As, a) \to \cdots \to \Mk (\As, a) \to \cdots \]
By taking the colimits of these diagrams, we obtain a comonad $\Momega$, which corresponds to the usual unfolding of a Kripke structure to all finite levels. 
Similar constructions can be applied to the other comonads. 

\textbf{Relations between fragments and parameters}
We can define morphisms between the different comonads we have discussed, which yield proofs about the relationships between the logical fragments they characterize. This categorical perspective avoids the cumbersome syntactic translations  in the standard proofs of these results. For illustration, there is a comonad morphism $t: \Ek \Rightarrow \Pk$ with components $t_A: \Ek \As \rarr \Pk \As$ given by $[a_{1},\dots,a_{j}] \mapsto [(1,a_{1}),\dots,(j,a_{j})]$. Together with  theorems \ref{thm:mainPebble} and \ref{thm:mainEF}, this shows that $\ELk \subseteq \ELvk$ and $\Lck \subseteq \Lvck$. Moreover, composing $t$ with a coalgebra $\As \rarr \Ek \As$ yields a coalgebra $A \rarr \Pk \As$, demonstrating that $\tw(\As) + 1 \leq \td(\As)$. 

Another morphism $\Momega \Rightarrow \PM_2$ shows that modal logic can be embedded into $2$-variable logic. More precisely, we can lift $\PM_2$ to $\CSp$ by defining the distinguished element of $\PM_2 (\As, a)$ to be $[(1,a)]$. The component of the morphism at $(\As, a)$ sends 
$[a_0, \alpha_1, a_1, \ldots , \alpha_j, a_{j}]$ to $[(1, a_0), (2, a_1), (1, a_2), \ldots ((j \mod 2) + 1, a_j)]$. This captures the ``hand-over-hand'' reuse of variables in a syntax-free fashion.

\subsection*{Concluding remarks}
Our comonadic constructions for the three major forms of model comparison games show a striking unity, on the one hand, but also some very interesting differences. For the latter, we note the different forms of logical ``deception'' associated with each comonad, 
the finite character of $\Ek$ and $\Mk$ and the non-finite character of $\Pk$,
and the different combinatorial parameters which arise in each case.

One clear direction for future work is to gain a deeper understanding of what makes these constructions work. Another is to understand how widely the comonadic analysis of resources can be applied. We are currently investigating the guarded fragment \cite{andreka1998modal,gradel2014freedoms}; other natural candidates include existential second-order logic, and branching quantifiers and dependence logic \cite{vaananen2007dependence}.

Since comonads arise naturally in type theory and functional programming \cite{uustalu2008comonadic,orchard2014programming}, can we connect the study of finite model theory made here with a suitable type theory? Can this lead, via the Curry-Howard correspondence, to the systematic derivation of some significant meta-algorithms, such as decision procedures  for guarded logics based on the tree model property \cite{gradel1999decision}, or  algorithmic metatheorems such as Courcelle's theorem \cite{courcelle1990monadic}?

Another intriguing direction is to connect these ideas with the graded quantum monad studied in \cite{abramskyquantum}, which provides a basis for the study of quantum advantage in $\CS$. This may lead to a form of quantum finite model theory.

\subsection*{Current and ongoing work}
The first author of the present paper and Anuj Dawar have recently begun a U.K. EPSRC-funded joint project on ``Resources and Co-Resources: a junction between categorical semantics and descriptive complexity'' to pursue the ideas introduced in \cite{abramsky2017pebbling} and the present paper. Project participants include the second author of the present paper, Tom Paine, Adam \'O Conghaile, Daniel Marsden, Tom\'a\v{s} Jakl and Luca Reggio.  A categorical analysis of Rossman's Equirank HPT Theorem appears in \cite{abramsky2020whither}, and a significant refinement of this result, involving a combination of the pebbling and Ehrenfeucht-\Fraisse comonads, is presented in \cite{Paine2020}.
Current work in progress includes comonadic treatments of guarded fragments \cite{abramsky2021comonadic} and of generalized quantifiers \cite{conghaile2020game}, and coalgebraic characterizations of further combinatorial invariants such as clique-width.

\bibliographystyle{amsplain}
\bibliography{bibfile}

\end{document}